\documentclass[journal]{IEEEtran}

\pdfoutput=1

\usepackage{cite}
\usepackage{amsmath,amsfonts,amssymb,amsthm,amsbsy}
\usepackage{enumerate}
\usepackage{booktabs}
\usepackage{color}
\usepackage{graphicx}
\usepackage{subfigure}
\usepackage[ruled,vlined,linesnumbered]{algorithm2e}
\usepackage{bbm}
\usepackage{diagbox}
\usepackage{makecell}

\newcommand{\mb}[1]{\pmb{#1}}
\newcommand{\Diag}[1]{\text{Diag}\left(#1\right)}

\newcommand\etal[0]{et al.}

\newcommand\refeq[1]{(\ref{#1})}
\newcommand\refalg[1]{Algorithm~{\ref{#1}}}
\newcommand\reffig[1]{Figure~{\ref{#1}}}
\newcommand\reftab[1]{Table~{\ref{#1}}}
\newcommand\refsec[1]{Section~{\ref{#1}}}
\newcommand\refapdx[1]{Appendix~{\ref{#1}}}
\newcommand\refprop[1]{Proposition.~{\ref{#1}}}
\newcommand\refthm[1]{Theorem.~{\ref{#1}}}

\newcommand\refcor[1]{Corollary.~{\ref{#1}}}
\newcommand\reflma[1]{Lemma.~{\ref{#1}}}

\theoremstyle{definition}
\newtheorem{proposition}{Proposition}
\newtheorem{corollary}{Corollary}
\newtheorem{theorem}{Theorem}
\newtheorem{lemma}{Lemma}

\theoremstyle{remark}

\usepackage[hidelinks]{hyperref}

\SetKwInput{KwParams}{Parameters}
\newcommand{\nonl}{\renewcommand{\nl}{\let\nl\oldnl}}
\begin{document}
\title{Opinion Control under Adversarial Network Perturbation: A Stackelberg Game Approach}

\author{Yuejiang~Li, Zhanjiang Chen,
        H.~Vicky~Zhao
\thanks{Y. Li, Z. Chen and H. V. Zhao was with the Department of Automation, Tsinghua University, Beijing, 10084 P. R. China e-mail: {lyj18,czj17}@mails.tsinghua.edu, vzhao@tsinghua.edu.cn.}
\thanks{Manuscript received April 19, 2005; revised August 26, 2015.}}

%
%

%

\maketitle

\begin{abstract}
  The emerging social network platforms enable users to share their own opinions, as well as to exchange opinions with others. However, adversarial network perturbation, where malicious users intentionally spread their extreme opinions, rumors and misinfomration to others, is ubiquitous in social networks. Such adversarial network perturbation greatly influences the opinion formation of the public, and threatens our societies. Thus, it is critical to study and control the influence of adversarial network perturbation. Although tremendous efforts have been made in both academia and industry to guide and control the public opinion dynamics, most of these works assume that the network is static, and ignore such adversarial network perturbation. In this work, based on the well-accepted Friedkin-Johnsen opinion dynamics model, we model the adversarial network perturbation, and analyze its impact on the networks' opinion. Then, from the adversary's perspective, we analyze its optimal network perturbation, which maximally changes the networks' opinion. Next, from the network defender's perspective, we formulate a Stackelberg game, and aim to control the network's opinion even under such adversarial network perturbation. We device a projected subgradient algorithm to solve the formualted Stackelberg game. Extensive simulations on real social networks validate our analysis of the advesarial network perturbation's influence and the effectiveness of the proposed opinion control algorithm.
\end{abstract}


\IEEEpeerreviewmaketitle

\section{Introduction}
The emerging online social network (OSN) platforms, such as Facebook, Twitter, and WeChat, etc., greatly strengthen the connections among people around the world and ultimately shaped the way of how people form their opinions.
People can easily share their own information anywhere and anytime, and they can also exchange opinions through comments, likes, or reposts.
In addition, users in social network platforms can not only interact with their own known friends. They can also exploit the recommendation function in OSN platforms, interact with those they did not know, and expand their circles all the time.

It is a controversial issue that people can gain other strangers information and opinions in such a convenient way.
On the one hand, governments can use OSN platforms to spread their public statements.
One the other hand, some malicious users and partisan media inject their extreme opinions \cite{guess2021consequences}, influence other innocent users, and trigger severe social riots. A recent example is the ``storming of U.S. Capitol'' in Jan. 2021. Inflamed by President Trump's speech and tweets, ``flag-waving, chanting and cursing throngs overwhelmed security barriers, and the whole building was lockdown \cite{editorial2021}.”
Thus, it is critical to investigate the influence of such extreme opinions, evaluate their impact on the OSN platforms, and further design effective defense mechanism to control the spread of such extreme opinions.

\subsection{Literature Review}
\textbf{Opinion Dynamics Models.} Modeling and analyzing opinion dynamics in social networks have received research attention from different disciplines.
In the literature, there were two lines of works that modeled and analyzed opinion dynamics in social networks.
The first studied the discrete-valued opinion scenario, 
A line of works built discrete opinion dynamics models and studied binary opinion scenario, for example, to support the Democrats or the Republicans. 
In the discrete models, users imitated their neighbors' opinions to update their own.
Popular imitation rules included random imitation rule in the voter model \cite{liggett2013stochastic}, the local majority rule \cite{krapivsky2003dynamics}, and the linear threshold rule \cite{granovetter1978threshold}, etc.
The imitation in the discrete models resulted in users' opinions to change to change from one extreme to the other.
To model how users opinions were changed and formed in a gradual process, the other continuous-valued opinion dynamics models were proposed.
In the continuous models, opinion values were in a certain range, with two endpoints representing two extreme opinions.

One of the fundamental works of the continuous model was the DeGroot model \cite{degroot1974reaching}, where users updated their opinions by averaging those of his/her neighbors.
In the DeGroot model, it would reach the opinion consensus state where all agents held the same opinion.
Hegselmann and Krause assumed that users would ignore opinions that were too far from theirs when updating opinions, and they found that users' opinions would converge to different clusters at equilibrium \cite{hegselmann2002opinion}.
The Friedkin-Johnsen (FJ) model incorporated users' intrinsic beliefs and their stubbornness into the DeGroot model, and it was shown in \cite{friedkin1990social,bindel2015bad} that such stubbornness could cause opinion polarization where users held different opinions at the equilibrium.
In addition, the FJ model was validated by small and medium group of social experiments \cite{friedkin2011formal}.

\textbf{Opinion Control.} Based on the above opinion dynamics models, there were many works studying how to control public opinion.
These works mainly focused on two objectives: the total opinion and the opinion polarization.
The total opinion characterized the overall stance of the whole population.
The work in \cite{kempe2003maximizing} first studied the problem of maximizing total opinion in the discrete opinion dynamics models.
They proposed a greedy algorithm to select a certain number of ``seed users'' that can maximized the total opinion.
The work in \cite{gionis2013opinion} further extended this algorithm to the continuous opinion dynamics models.
The total opinion was optimized by controlling users' stubbornness in \cite{abebe2018opinion} and by manipulating users' intrinsic beliefs through persuasion in \cite{xu2020opinion}.
Opinion polarization quantified the discrepancies among users’ opinions.
The work in \cite{matakos2017measuring} minimized the opinion polarization by controlling users' intrinsic believes.
Furthermore, the works in \cite{chen2018quantifying} and \cite{musco2018minimizing} studied how network structure influence the opinion polarization, respectively.
Chen \etal. studied how to maximize the opinion polarization from the adversary's perspective \cite{chen2021adversarial,gaitonde2020adversarial}.
Although the above algorithms could efficiently control the total opinion or the opinion polarization, they were developed based on the assumption that the network structure was static.
That is, the connections and influence strength among all users did not change during the opinion dynamics process.

\textbf{Dynamic Network Structure.} In reality, the network structure may change from time to time \cite{greene2010tracking,pereira2016evolving}.
The work in \cite{barabasi1999emergence} proposed the preferential attachment algorithm that the newly added user were more likely to follow those with more followers. The preferential attachment could explain the power-law degree distribution of social networks in reality.
Some works explored the data from social network platforms and analyzed how real social networks changed over time.

The work in \cite{chitra2020analyzing} showed that the recommending algorithms in social network platforms tended to connect users with similar interests together.
This could further resulted in the ``filter bubble'' effect \cite{pariser2011filter} that the opinion polarization was increased.
The work in \cite{ceci2018small} proposed perturbation centrality based on the perturbation analysis of the network structure.
Furthermore, the perturbation centrality was used in graph-based optimization problem, which was proved to be robust against edge failure.

Different from the prior works on opinion control, which assumed that the network structure is static, in this work, we consider the scenario where some attackers can perturb the network structure by spread their extreme opinions to other target users. We model such adversarial network perturbation. From the adversary's perspective, we aim to analyze the optimal strategy for the adversary to choose the attackers and the corresponding target users. Then, from the network defender's perspective, we device a opinion control scheme which is robust to the adversarial network perturbation.
Our investigation is important to the development of robust online social networks.

\subsection{Our Contribution}
Our contributions can be summarized as follows:
\begin{enumerate}
    \item We model the network perturbation and theoretically analyze its impact on the total opinion.
    \item We consider the scenario where the adversary can exploit the network perturbation to maximize the total opinion, while the network defender aims to minimize the total opinion under such adversarial network perturbation. We model such scenario as a Stackelberg game, played between the network defender and the adversarial.
    \item From the adversary's perspective, we theoretically analyze the optimal strategy to choose the attacker and the corresponding target users, so that the total opinion is maximized.
    \item From the network defender's perspective, we formulate a min-max game which aims to minimize the total opinion, even if the network is perturbed by the adversary with their optimal strategy. We further device a project subgradient method to efficiently solve the formulated min-max game.
    \item We conduct extensive experiments on real social networks to validate the analysis of the network perturbation, and the optimal strategy for the adversary. We also use real social networks to test the defense algorithms.
\end{enumerate}

The rest of this paper is organized as follows. In \refsec{sec:prelim}, we introduce the basics of the Friedkin-Johnsen opinion dynamics models and opinion control. In \refsec{sec:perturbation}, we model the network perturbation, and formulate the Stackelberg game played between the adversary and the network defender. Then, in \refsec{sec:attack}, we analyze the optimal strategy for the adversary to maximize the total opinion. Next, in \refsec{sec:defense}, from the defender's perspective, we solve the stackelberg game to minimize the total opinion under the adversarial network perturbation.
The simulation and experiments results are summarized in \refsec{sec:exp}, and the conclusion is drawn in \refsec{sec:conclusion}.
\section{Preliminary}
\label{sec:prelim}
In this section, we briefly review the Friedkin-Johnsen (FJ) opinion dynamics model, and introduce the opinion minimization problem studied in this paper. 

\subsection{The Friedkin-Johnsen Model}
In this work, we consider the scenario that there are $n$ users in network, discussing a topic that is harmful to public security.
The network can be modeled as a graph, where nodes represent users, and an directed edge $(i,j)$ indicates that user-$j$ can influence user-$i$.
Let $\mb{W}$ be the adjacency matrix of the graph, whose entry $W_{ij} > 0$ shows the influence weight of user $j$ on user $i$ if there is an directed edge $(i, j)$; otherwise, $W_{ij} = 0$.
Following the original FJ model \cite{friedkin1990social}, for any user $i$, the total influence weight of all other users on user $v$ are unified to 1. That is, $\mb{W}\mb{1} = \mb{1}$.

In the FJ model, user $i$ holds internal opinion $s_i$ which shows his/her intrinsic belief on the discussed topic.
Let $\mb{s} = [s_1, \cdots, s_n]^T$ be all users' internal opinions.
The opinion formation process is divided into discrete time steps.
At time step $t$, the expressed opinions of users are $\mb{z}(t) = [z_1(t), \cdots, z_n(t)]^T$.
Each user aggregates his/her neighbors' expressed opinions together with his/her own internal opinion, and updates the expressed opinion at the next time step as
\begin{equation}
\label{eqn:fj-dyn-one-user}
    z_i(t+1) = \alpha_i s_i + (1 - \alpha_i) \cdot \sum_{j} W_{ij} z_j(t),
\end{equation}
where $0< \alpha_i < 1$ is the stubbornness of user $i$. The larger $\alpha_i$ is, user $i$ is more stubborn and follows his/her own intrinsic belief more.

The expressed opinions of the whole population evolve as the above updating process, and finally reach the equilibrium state \cite{friedkin1990social}
\begin{gather}
    \mb{z}^* = \mb{B} \mb{A} \mb{s},
    \label{eqn:fj-equilibrium}
    \text{ where }  \\
    \mb{B} = [\mb{I} - (\mb{I} - \mb{A})\mb{W}]^{-1} \text{ and } \mb{A} = \Diag{\alpha_i}.\label{eqn:B}
\end{gather}
From \refeq{eqn:fj-equilibrium}, the equilibrium expressed opinions depends on the network structure $\mb{W}$, all users internal opinions $\mb{s}$, and their stubbornness $\mb{A}$.

\subsection{Opinion Control}
As the total opinion shows the whole population's support of the discussed topic, following the previous works in \cite{xu2020opinion,abebe2018opinion,chan2019revisiting}, we aim to minimize the total opinion at equilibrium, that is
\begin{equation}
    f = \mb{1}^T \mb{z}^* = \mb{1}^T \mb{B} \mb{A} \mb{s}
    \label{eqn:ori-total-opinion}
\end{equation}
In the following, we refer to $f$ as total opinion for short when no confusions are made.
We adopt the ``Min-Total'' algorithms in \cite{xu2020opinion} as the baseline methods.
In the ``Min-Total'' algorithm, users internal opinions can be controlled through e.g., persuasion, but limited to a certain budget $\mu$. Consequently, the total opinion minimization problem can be formulated as
\begin{equation}
\begin{aligned}
    \min_{\mb{x}}   & \quad \mb{1}^T \mb{B} \mb{A} \mb{x}\\
    s.t.            & \quad 
    \left\{\begin{array}{lc}
      \mb{1}^T (\mb{s} - \mb{x}) \le \mu,\\
      \quad \mb{0} \le \mb{x} \le \mb{s}.
    \end{array}\right.
\end{aligned}
\label{eqn:opinion-minimization}
\end{equation}
The formulated minimization is convex, and can be solved efficiently using linear searching algorithm \cite{xu2020opinion}.
The above algorithm can obtain the optimal control when the network structure $\mb{W}$ is static and unchanged. However, both evidence in literature \cite{barabasi1999emergence} and real data \cite{greene2010tracking,pereira2016evolving} have shown that the network structure may change during the opinion formation process. In the following, we model the change of network structure, and analyze its impact on the total opinion.
\section{Adversarial Network Perturbation and Problem Formulation}
\label{sec:perturbation}
In online social network platforms, a user $v$ can easily explore other users opinion, e.g., trending tweets on the exploration page of Twitter and private messages from strangers on Instagram, even though they are not $v$'s neighbors.
Some malicious users may exploit this nature of online social network platforms and intentionally spread their own ideas to the public.
We define this intentional change in network structure as \text{adversarial network perturbation}.
In this section, we model the adversarial network perturbation, and formulate the robust opinion minimization problem under such perturbation.

\subsection{Adversarial Network Perturbation Model}
\label{sec:net-perturb}
In this work, we consider the scenario where the adversary aims to promote the support of the discussing harmful topic, and try to maximize the total opinion of the network.
We assume that the adversary can manipulate at most $m$ users as the \text{attackers}, and push the expressed opinions of these attackers to some target users.
Let $\mathcal{A}$ and $\mathcal{T}$ be the set of attackers and that of target users, respectively.
For an attacker $u\in \mathcal{A}$, the target user of him/her is denoted by $\mathcal{T}_u$.
In this work, we assume that each attacker can push his/her expressed opinions to at most $k$ target users, that is, $\vert \mathcal{T}_u \vert \le k ~ \forall u\in \mathcal{A}$.
For a target user $v \in \mathcal{T}$, let $\mathcal{A}_v$ be the attackers that influences this target user. Consequently, we have
\begin{equation}
    \begin{aligned}
    &\bigcup_{u \in \mathcal{A}} \mathcal{T}_u = \mathcal{A}, \text{ and}
    &\bigcup_{v \in \mathcal{T}} \mathcal{A}_v = \mathcal{T}.
    \end{aligned}
    \label{eqn:set-relation}
\end{equation}

For a target user $v\in \mathcal{T}$, when an attacker $u\in \mathcal{A}_v$ pushes his/her expressed opinion to $v$, we assume that it influences the target user $v$ with weight $p$. We also define $p$ as the perturbation coefficient. To ensure that the total influence weights of other users on the target user $v$ is unified to 1, the influence weights of $v$'s original neighbors is discounted by $(1 - \vert \mathcal{A}_v \vert p)$, and $\vert \mathcal{A}_v \vert p \le 1$. Consequently, the influence weights on target user $v$ becomes
\begin{equation}
    \tilde{W}_{vu} = \begin{cases}
        (1 - \vert \mathcal{A}_v \vert p)\cdot W_{vu} & u \notin \mathcal{A}_v;\\
        (1 - \vert \mathcal{A}_v \vert p)\cdot W_{vu} + p & u \in \mathcal{A}_v.
    \end{cases}
\end{equation}
For other users $\bar{v} \notin \mathcal{T}$, others users' influence weights on him/her do not change.
The adjacency matrix after perturbation becomes
\begin{equation}
    \tilde{\mb{W}} = \mb{W} + p \times \mb{\Delta}_{W},\text{ where}
    \label{eqn:new-W}
\end{equation}
\begin{equation}
    \mb{\Delta}_{W} = \sum_{v \in \mathcal{T}} \mb{e}_v \sum_{u \in \mathcal{A}_v} \mb{e}_u^T - \sum_{v \in \mathcal{T}} \vert \mathcal{A}_v\vert \mb{e}_v \mb{e}_v^T\mb{W}.
    \label{eqn:delta-w}
\end{equation}

Since the adjacency matrix has changed to $\tilde{\mb{W}}$, with \refeq{eqn:fj-equilibrium}, the expressed opinion at equilibrium becomes
\begin{equation}
    \tilde{\mb{z}^*} = \tilde{\mb{B}} \mb{A} \mb{s}, \text{ where } \tilde{\mb{B}} = (\mb{I} - (\mb{I} - \mb{A})\tilde{\mb{W}})^{-1},
    \label{eqn:perturb-opinion}
\end{equation}
and the total opinion becomes
\begin{equation}
    \tilde{f} = \mb{1}^T \tilde{\mb{B}} \mb{A} \mb{s}.
\end{equation}

\subsection{Network Defense under Adversarial Perturabtion}
Owing to the existence of the potential adversarial network perturbation, in this work, we formulate the following \textit{network defense} problem.
\begin{equation}
  \begin{aligned}
    \min_{\mb{x}} & \quad\max_{\mathcal{A}, \mathcal{T}_u} \mb{1}^T \tilde{\mb{B}} \mb{A} \mb{x}\\
    s.t.          & \quad 
    \left\{\begin{array}{lc}
      \mb{1}^T (\mb{s} - \mb{x}) \le \mu,\\
      \mb{0} \le \mb{x} \le \mb{s},\\
      \vert \mathcal{A} \vert \le m,\\
      \vert \mathcal{T}_u \vert \le k,~ \forall u\in \mathcal{A}.
    \end{array}\right.
  \end{aligned}
  \label{eqn:rocpi}
\end{equation}

The formulated network defense problem in \refeq{eqn:rocpi} can be interpreted as a zero-sum game played between the adversary and the defender.
The adversary can perturb the network as analyzed in \refsec{sec:net-perturb} to maximize the total opinion.
This corresponds to the inner maximization of the objective function in \refeq{eqn:rocpi}.
For the defender, similar to the prior works in \cite{xu2020opinion,musco2018minimizing}, we assume that they control users' innate opinions with a certain budget to minimize the total opinion.
Different from the opinion minimization problem in  \refeq{eqn:opinion-minimization}, here, the defender is aware of the existence of the adversary, and assume that the adversary is rational to maximize the total opinion. Thus, the objective of the defender in \refeq{eqn:rocpi} is the maximized total opinion by the adversary.

To solve the proposed network defense problem under the adversarial network perturbation, in the following, we first analyze from the adversary's perspective to find the optimal solution to the inner maximization of \refeq{eqn:rocpi}.
Then, based on the analytical solution to the inner maximization problem, we further develop efficient algorithm to solve \refeq{eqn:rocpi}.

\section{Optimal Adversarial Network Perturbation}
\label{sec:attack}
In this section, from the adversary's perspective, given that the controlled internal opinions are $\mb{s}$, we study how to maximize the total opinion by selecting attackers $\mathcal{A}$ and the corresponding target users $\mathcal{T}_u$ for $u\in \mathcal{A}$.
That is, we solve the following problem
\begin{equation}
    \begin{aligned}
    \max_{\mathcal{A}, \mathcal{T}_u} &\quad\quad\mb{1}^T \tilde{\mb{B}} \mb{A} \mb{s}\\
    s.t.          & \quad 
    \left\{\begin{array}{lc}
      \vert \mathcal{A} \vert \le m,\\
      \vert \mathcal{T}_u \vert \le k,~ \forall u\in \mathcal{A}.
    \end{array}\right.
  \end{aligned}
  \label{eqn:maximization}
\end{equation}
The first constraint in the above shows that the adversary can select up to $m$ attackers, while the second constraints indicates that each attacker $u\in\mathcal{A}$ can perturb up to $k$ target users.

Note that the above maximization problem is a combinatorial problem.
In addition, the optimization variables $\mathcal{A}$ and $\mathcal{T}_u$ are included in $\tilde{\mb{B}}$. From \refeq{eqn:perturb-opinion}, the outer inverse further hinder the solution of the maximization.
A direct method is to use the greedy algorithm as in \cite{kempe2003maximizing}, to iteratively select the attackers and the target users.
However, computational cost is substantial, and it is unacceptable when the network size is sufficiently large.

In the following of this section, we first approximate the objective function in \refeq{eqn:maximization}.
Then, based on the approximated objective function, we develop a linear search algorithm to solve the approximated maximization problem.

\subsection{Approximation of Objective}
Note that in real social network, users are often less likely to be influenced by users who are not their friends. 
Consequently, the impact of the adversarial network perturbation is limited.
In our model, this corresponds to the perturbation coefficient $p$ is sufficient small, that is, $p \rightarrow 0$.  Based on the assumption that $p \rightarrow 0$, we have the following proposition.
\begin{proposition}
  \label{prop:approx-new-total}
  when the network structure changed to $\tilde{\mb{W}}$, the total opinion with perturbation  can be approximated as
  \begin{equation}
    \tilde{f} \approx \mb{1}^T \mb{B} \mb{A} \mb{s} + p\times \mb{1}^T \mb{B} (\mb{I} - \mb{A}) \mb{\Delta}_{W} \mb{B} \mb{A} \mb{s}
    \label{eqn:approx-new-total}
  \end{equation}
\end{proposition}
\begin{proof}
  Using the Taylor series, we have
  $$
  \tilde{f} = \left.\tilde{f}\right\vert_{p=0} + p \times \left.\frac{d \tilde{f}}{dp} \right\vert_{p=0} + \mathcal{O}(p^2).
  $$
  When $p = 0$, with $\tilde{\mb{W}}$ in \refeq{eqn:new-W}, $\tilde{f} = f = \mb{1}^T \mb{B} \mb{A} \mb{s}$. To calculate the second term in the above, we first calculate $\frac{d \tilde{f}}{dp}$.
  $$
    \begin{aligned}
      \frac{d \tilde{f}}{dp} &= \frac{d(\mb{1}^T \tilde{\mb{B}} \mb{A} \mb{s})}{dp} = \mb{1}^T \frac{d \tilde{\mb{B}}}{dp} \mb{A} \mb{s}.
    \end{aligned}
  $$
  As $\tilde{\mb{B}} = (\mb{I} - (\mb{I} - \mb{A})\tilde{\mb{W}})^{-1}$, and we have $\tilde{\mb{B}}\cdot (\mb{I} - (\mb{I} - \mb{A})\tilde{\mb{W}}) = \mb{I}$. Taking derivatives w.r.t. $p$ on both sides, we have
  $$
  d\tilde{\mb{B}} (\mb{I} - (\mb{I} - \mb{A})\tilde{\mb{W}}) - \tilde{\mb{B}} (\mb{I} - \mb{A}) d \tilde{\mb{W}} = \mb{0}.
  $$
  With $\tilde{\mb{W}}$ in \refeq{eqn:delta-w}, we have $d \tilde{\mb{W}} = \mb{\Delta}_{W} dp$. Then, we have
  $$
  \begin{aligned}
    & d\tilde{\mb{B}} (\mb{I} - (\mb{I} - \mb{A})\tilde{\mb{W}}) - \tilde{\mb{B}} (\mb{I} - \mb{A}) \mb{\Delta}_{W} dp = \mb{0}\\
    \Rightarrow &\frac{d\tilde{\mb{B}}}{dp} = \tilde{\mb{B}} (\mb{I} - \mb{A}) \mb{\Delta}_{W} (\mb{I} - (\mb{I} - \mb{A})\tilde{\mb{W}})^{-1} \\
    &\hspace{15pt}= \tilde{\mb{B}} (\mb{I} - \mb{A}) \mb{\Delta}_{W} \tilde{\mb{B}}
  \end{aligned}
  $$
  With the above derivatives we have
  $$
  \begin{aligned}
    \left.\frac{d \tilde{f}}{dp} \right\vert_{p=0} = \mb{1}^T \left(\frac{d \tilde{\mb{B}}}{dp}\right)_{p=0} \mb{A} \mb{s}
    = \mb{1}^T \mb{B} (\mb{I} - \mb{A}) \mb{\Delta}_{W} \mb{B} \mb{A} \mb{s}.
  \end{aligned}
  $$
  Ignoring the $\mathcal{O}(p^2)$ term, we have \refeq{eqn:approx-new-total}
\end{proof}

Note that the first term in \refeq{eqn:approx-new-total} is the total opinion $f$ when there is no adversarial network perturbation. Therefore, the total opinion with network perturbation can be regarded as $f$ plus a perturbation term related to how network structure is changed, i.e., $\mb{\Delta}_{W}$. The contribution of the perturbation term to the total opinion is controlled by the perturbation coefficient $p$.
When $p$ is larger, the total opinion at equilibrium deviates more from the original one.

\subsection{Optimal Selection of the Attackers and the Target Users}
\label{sec:opt-attack}
With the above approximation, in this section, we derive the optimal adversarial network perturbation.
That is, we find the optimal set of attackers $\mathcal{A}^*$. Furthermore, for each optimal attacker $u\in\mathcal{A}^*$, we find $u$'s optimal target user set $\mathcal{T}_u^*$.

In \refeq{eqn:approx-new-total}, the first term is neither related to the selection of attackers nor target users. Thus, the maximization in \refeq{eqn:maximization} is equivalent to maximizing the perturbation term, that is,
\begin{equation}
  \begin{aligned}
    \max_{\mathcal{A}, \mathcal{T}_u} &\quad \mb{1}^T \mb{B} (\mb{I} - \mb{A}) \mb{\Delta}_{W} \mb{B} \mb{A} \mb{s}\\
    s.t.          & \quad 
    \left\{\begin{array}{lc}
      \vert \mathcal{A} \vert \le m,\\
      \vert \mathcal{T}_u \vert \le k,~ \forall u\in \mathcal{A}.
    \end{array}\right.
  \end{aligned}
  \label{eqn:max-perturb-term}
\end{equation}

Then, with the change of adjacency matrix in \refeq{eqn:delta-w}, the above objective can be further written as
\begin{align}
    &\mb{1}^T \mb{B} (\mb{I} - \mb{A}) \mb{\Delta}_{W} \mb{B} \mb{A} \mb{s} \\
    = &\mb{1}^T \mb{B} (\mb{I} - \mb{A}) \left( \sum_{v \in \mathcal{T}} \mb{e}_v \sum_{u \in \mathcal{A}_v} \mb{e}_u^T - \sum_{v \in \mathcal{T}} \vert \mathcal{A}_v\vert \mb{e}_v \mb{e}_v^T \mb{W} \right) \mb{B} \mb{A} \mb{s}\nonumber\\
    =& \sum_{v \in \mathcal{T}} \mb{1}^T \mb{B} (\mb{I} - \mb{A}) \mb{e}_v \sum_{u \in \mathcal{A}_v} \mb{e}_u^T \mb{B} \mb{A} \mb{s}\\ 
    &- \sum_{v \in \mathcal{T}} \vert \mathcal{A}_v\vert \mb{1}^T \mb{B} (\mb{I} - \mb{A}) \mb{e}_v \mb{e}_v^T\mb{W} \mb{B} \mb{A} \mb{s}\nonumber
\end{align}
Note that from \refeq{eqn:fj-equilibrium}, we have $\mb{z}^* = \mb{B}\mb{A}\mb{s}$. We further define
\begin{equation}
    \begin{aligned}
        \mb{c}_1 &\triangleq (\mb{1}^T \mb{B} (\mb{I} - \mb{A}))^T, \text{ and}\\
        \mb{c}_2 &\triangleq \mb{W} \mb{z}^* = \mb{W} \mb{B} \mb{A} \mb{s}.
    \end{aligned}
    \label{eqn:defined-vectors}
\end{equation}
Consequently, the above objective can be further simplified as
\begin{align}
  &\mb{1}^T \mb{B} (\mb{I} - \mb{A}) \mb{\Delta}_{W} \mb{B} \mb{A} \mb{s}\\
  = &\sum_{v \in \mathcal{T}} c_1(v) \sum_{u \in \mathcal{A}_v} z^*(u) - \sum_{v \in \mathcal{T}} \vert \mathcal{A}_v\vert c_1(v) c_2(v)\nonumber\\
  =& \sum_{v \in \mathcal{T}} c_1(v)\cdot \left(\sum_{u \in \mathcal{A}_v} z^*(u) - \vert \mathcal{A}_v \vert c_2(v)\right)\nonumber\\
  =& \sum_{v \in \mathcal{T}} c_1(v)\cdot \sum_{u \in \mathcal{A}_v} \big(z^*(u) - c_2(v)\big)\nonumber\\
  =& \sum_{v \in \mathcal{T}} \sum_{u \in \mathcal{A}_v} c_1(v)\cdot \big(z^*(u) - c_2(v)\big)\label{eqn:perturb-target-view}\\
  =&\sum_{u \in \mathcal{A}} \sum_{v \in \mathcal{T}_u} c_1(v)\cdot \big(z^*(u) - c_2(v)\big).
  \label{eqn:perturb-attack-view}
\end{align}
From \refeq{eqn:perturb-target-view} to \refeq{eqn:perturb-attack-view}, we use the relationships in \refeq{eqn:set-relation} and exchange the summation order.
From \refeq{eqn:perturb-attack-view}, when an attacker $u$ influences a target user $v$, such perturbation contributes $\delta_{u,v} \triangleq c_1(v)\cdot \big(z^*(u) - c_2(v)\big)$ to the change of total opinion. Thus, we define $\delta_{u,v}$ as the \textit{meta-influence}.
Furthermore, for the attacker $u$, his/her total influence on the change of total opinion is
\begin{equation}
  \delta_u \triangleq \sum_{v \in \mathcal{T}_u} \delta_{u,v} =\sum_{v \in \mathcal{T}_u} c_1(v)\cdot \big(z^*(u) - c_2(v)\big)
  \label{eqn:attacker-influence}
\end{equation}
We also call $\delta_u$ the \textit{individual influence} of user $u$, if he/she is selected as the attacker.

With the derivation in \refeq{eqn:perturb-attack-view}, the maximization in \refeq{eqn:max-perturb-term} can be transformed as
$$
  \begin{aligned}
    &\max_{\mathcal{A}, \mathcal{T}_u} \quad \mb{1}^T \mb{B} (\mb{I} - \mb{A}) \mb{\Delta}_{W} \mb{B} \mb{A} \mb{s}\\
    =&\max_{\mathcal{A}, \mathcal{T}_u} \sum_{u \in \mathcal{A}} \sum_{v \in \mathcal{T}_u} c_1(v)\cdot \big(z^*(u) - c_2(v)\big)\\
    =&\max_{\mathcal{A}} \sum_{u \in \mathcal{A}} \left(\max_{\mathcal{T}_u}\sum_{v \in \mathcal{T}_u} c_1(v)\cdot \big(z^*(u) - c_2(v)\big)\right)\\
    =&\max_{\mathcal{A}} \sum_{u \in \mathcal{A}} \max_{\mathcal{T}_u} \delta_u.
  \end{aligned}
$$
To this end, the maximization in \refeq{eqn:max-perturb-term} can be decoupled as a two-step maximization problem.
The first step is to maximize the user $u$'s individual influence and choose his/her corresponding target user set $T^*_u$, if he/she is selected as the attacker. That is,
\begin{equation}
  \begin{aligned}
    \max_{\mathcal{T}_u} &\quad \delta_u = \sum_{v \in \mathcal{T}_u} c_1(v)\cdot \big(z^*(u) - c_2(v)\big)\\
    s.t.          & \quad \vert\mathcal{T}_u \vert \le k.
  \end{aligned}
  \label{eqn:max-attacker-influence}
\end{equation}
By solving \refeq{eqn:max-attacker-influence} for each user $u$, we obtain his/her optimal target user set $\mathcal{T}_u^*$ as well as his/her optimal individual influence $\delta_u^*$. Then, the second step is to solve
\begin{equation}
  \begin{aligned}
    \max_{\mathcal{A}} &\quad \sum_{u \in \mathcal{A}} \max_{\mathcal{T}_u} \delta_u = \sum_{u \in \mathcal{A}} \delta_u^*\\
    s.t.          & \quad \vert \mathcal{A} \vert \le m,
  \end{aligned}
  \label{eqn:max-total-influence}
\end{equation}
to obtain the optimal attacker set $\mathcal{A}^*$. In the following, we analyze the solutions to \refeq{eqn:max-attacker-influence} and \refeq{eqn:max-total-influence}.

\subsubsection{Optimal Selection of Target User}
According to the above analysis, selecting the users with the largest meta-influence $c_1(v)\cdot (z^*(u) - c_2(v))$ as the target user $T_u$ can maximize the individual influence $\delta_u$.
Recall the definition that $c_1 = (\mb{1}^T \mb{B} (\mb{I} - \mb{A}))^T$. Since $0 < \alpha_i < 1$, the diagonal of $(\mb{I} - \mb{A})$ is always positive.
From \refeq{eqn:B}, using the matrix identity, we have
\begin{equation}
  \mb{B} = [\mb{I} - (\mb{I} - \mb{A})\mb{W}]^{-1} = \mathbf{I} + (\mb{I} - \mb{A})\mb{W} + \big((\mb{I} - \mb{A})\mb{W}\big)^2 + \cdots.
\end{equation}
Consequently, all entries of $\mb{B}$ and $\mb{1}^T \mb{B}$ are positive. Thus $c_1(v)$ is always positive.
Let $\hat{b}_i$ be the sum of the $i$-th column of $\mb{B}$. We can have $c_1(v) = \hat{b}_v \cdot (1 - \alpha_v)$.
Therefore, if user $v$ is easily influenced by others and has a smaller $\alpha_v$, he/she can have a larger value of $c_1(v)$, and thus, be selected as the target user.

Next, we analyze the second term $(z^*(u) - c_2(v))$. With \refeq{eqn:defined-vectors}, we can see that $c_2(v) = \sum W_{vj} z^*(j)$ is the weighted average of user $v$'s original neighbors opinions. If the neighbors of user $v$ have smaller expressed opinions (smaller $c_2(v)$), user $v$ is more likely to be selected as the target user by the attackers. Note that if a user $u$ is selected as the attacker, the size of his/her optimal target users $\vert \mathcal{T}_u^*\vert$ may be less than the constraint $k$. That is, selecting more target user can not always increase the individual influence $\delta_u^*$. This is because that $z^*(u)$ can be less than $c_2(v)$. Consequently, the term $(z^*(u) - c_2(v))$ and the meta-influence $c_1(v)\cdot (z^*(u) - c_2(v))$ can be negative. In this sense, the individual influence $\delta_u$ may decrease if we select this user $v$ as target user.

An extreme case is the user $\underline{u}$, who own the smallest expressed opinion when there is no adversarial perturbation, that is, $z^*(\underline{u}) \le z^*(u),~ \forall u$. We can have $(z^*(\underline{u}) - c_2(v)) \le 0$. Thus, user $\underline{u}$ can never be selected as the attacker.
The other extreme case is the user $\overline{u}$, who owns the largest expressed opinion when there is no adversarial perturbation, that is, $z^*(\overline{u}) \ge z^*(u),~ \forall u$. We can have $(z^*(\overline{u}) - c_2(v)) \ge 0$. Therefore, if user $\overline{u}$ will always select $k$ target users, it he/she is selected as the attacker. In fact, as we can see next, user $\overline{u}$ is always selected as the attacker.


\subsubsection{Optimal Selection of Attacker} Given the optimal target user selection criterion above, we can decide the optimal target user set $\mathcal{T}^*_u$ and the optimal individual influence $\delta_u^*$ for each user $u$. Then, we can select the users with the largest optimal individual influence $\delta_u^*$ as the optimal attackers $\mathcal{A}^*$ to solve \refeq{eqn:max-total-influence}. Note that, through this method, we first need to solve \refeq{eqn:max-attacker-influence} for each user, and the complexity is $\mathcal{O}(N\cdot N \log k)$. Next, we need to solve \refeq{eqn:max-total-influence} with complexity $N \log m$. Then, the total computational complexity is $\mathcal{O}(N\cdot (N \log k + \log m))$. We next analyze the properties of the selected attackers, and reduce the computtaional complexity.

\begin{proposition}
  \label{prop:choose-attacker}
  For a pair of users users $u_1, u_2 \in \mathcal{V}$, if $z^*(u_1) \ge z^*(u_2)$, then we have $\delta_{u_1}^* \ge \delta_{u_2}^*$. 
\end{proposition}
\begin{proof}
  Let $\mathcal{T}_{u_1}^*$ and $\mathcal{T}_{u_2}^*$ be the optimal target user set of $u_1$ and $u_2$, respectively. Since $\mathcal{T}_{u_1}^*$ is the optimal target user set, we have
  $$
    \begin{aligned}
      \delta_{u_1}^* &= \sum_{v \in \mathcal{T}^*_{u_1}} c_1(v)\cdot \big(z^*(u_1) - c_2(v)\big)\\
        &\ge \sum_{v \in \mathcal{T}^*_{u_2}} c_1(v)\cdot \big(z^*(u_1) - c_2(v)\big).
    \end{aligned}
  $$
  Note that $z^*(u_1) \ge z^*(u_2)$ and $c_1(v) \ge 0$ from the previous section, we have
  $$
    \begin{aligned}
      \delta_{u_1}^* &\ge \sum_{v \in \mathcal{T}^*_{u_2}} c_1(v)\cdot \big(z^*(u_1) - c_2(v)\big)\\
      &\ge \sum_{v \in \mathcal{T}^*_{u_2}} c_1(v)\cdot \big(z^*(u_2) - c_2(v)\big) = \delta_{u_2}^*.
    \end{aligned}
  $$
  Here, we have $\delta_{u_1}^*\ge \delta_{u_2}^*$.
\end{proof}
From \refprop{prop:choose-attacker}, we can see that if a user with a larger expressed opinion $z^*(u)$ withou perturbation, he/she can have a larger optimal individual influence, even without the knowledge of his/her optimal target user set. Consequently, we can have the following corollary.
\begin{corollary}
  \label{cor:opt-attacker-set}
  Let $\mathcal{A}_{[m]}$ be the set of top-$m$ users who own the largest expressed opinion without perturbation. That is, $\mathcal{A}_{[m]} = \{u_{[1]}, u_{[2]}, \cdots, u_{[m]}\}$, and $z^*([1]) \ge z^*([2]) \ge \cdots \ge z^*([m]) \ge z^*(v)$, for any $v\notin \mathcal{A}_{[m]}$. Then, $\mathcal{A}^* \subseteq \mathcal{A}_{[m]}$.
\end{corollary}
\begin{proof}
  From \refprop{prop:choose-attacker}, if $z^*([1]) \ge z^*([2]) \ge \cdots \ge z^*([m]) \ge z^*(v)$, we have $\delta^*_{u_{[1]}} \ge \delta^*_{u_{[2]}} \ge \cdots \ge \delta^*_{u_{[m]}} \ge \delta^*_v$, for any $v \notin \mathcal{A}_{[m]}$. Note that from the analysis in the above section, a user $u^\prime \in \mathcal{A}_{[m]}$ can have negative individual influence $\delta_{u^\prime}^*$ if $\big(z^*(u^\prime) - c_2(v) < 0$ for all $v\in \mathcal{V}$. In this case, $u^\prime$ should not be included in $\mathcal{A}^*$. Therefore, $\mathcal{A}^* \subseteq \mathcal{A}_{[m]}$.
\end{proof}
With \refprop{prop:choose-attacker} and \refcor{cor:opt-attacker-set}, we can simplify the searching process by first deciding the candidates of attackers $\mathcal{A}_{[m]}$ with the sorted expressed opinion $z^*(u)$. The complexity of this step is $\mathcal{O}(N \log m)$.Then, we only need to search for the target user set of each candidate attacker, and decide their optimal individual influence with \refeq{eqn:max-attacker-influence}. The complexity of this step is $\mathcal{O}(m \times N \log k)$. The overall complexity is $\mathcal{O}(N \log m + m \times N \log k) = \mathcal{O}(N \times(\log m + m \log k))$, which is significantly smaller than $\mathcal{O}(N\cdot (N \log k + \log m))$, when $m$ and $k$ are not comparable to $N$.
The algorithm for optimal attacking strategy is summarized in \refalg{alg:linear-search}.

\begin{algorithm}[t]
  \SetAlgoLined
  \KwIn{$\mb{c}_1, \mb{c}_2, \mb{z}^*$ and constraints of $m$ and $k$}
  \KwOut{Optimal attacker set $\mathcal{A}^*$ and optimal target user set $\mathcal{T}^*_u$ for $u\in \mathcal{A}^*$}
      Sorted $z^*(u)$ in descending order an pick top $m$ of it as the candidate attackers $\mathcal{A}_{[m]}$\;
      $\mathcal{A}^* \leftarrow \emptyset$\;
      \For(Searching for optimal attackers.){$u \in \mathcal{A}_{[m]}$}{
          Calculate $\delta_{uv} = c_1(v)\cdot \big(z^*(u) - c_2(v)\big)$ for each $v$\;
          Sorted $\delta_{uv}$ in descending order and pick top $k$ of it as $\mathcal{T}_u$\;
          $\mathcal{T}_u^* \leftarrow \emptyset$, $\delta_u^* \leftarrow 0$\;
          \For(Searching for optimal target users.){$v \in \mathcal{T}_u$}{
              \eIf{$\delta_{uv} > 0$}{
                $\delta_u^* \leftarrow \delta_u^* + \delta_{uv}$, $\mathcal{T}_u^* \leftarrow \mathcal{T}_u^* \cup \{v\}$                  
              }{
                  \texttt{break}\;
              }
          }
          \If {$\delta_u^* \le 0$}{
              \texttt{break}\;
          }
          $\mathcal{A}^* \leftarrow \mathcal{A}^* \cup \{u\}$\;
      }
      \caption{Linear search algorithm for \refeq{eqn:max-perturb-term}}
  \label{alg:linear-search}
\end{algorithm}

\section{Network Defense}
\label{sec:defense}
Based on the above discussion of the optimal attacking by the adversary, in this section, we consider how to minimize the total opinion in \refeq{eqn:rocpi} under such adversarial network perturbation. Here, we also relax the objective function with \refeq{eqn:approx-new-total} as in \refsec{sec:attack}, and formulate the following defense problem.
\begin{equation}
  \begin{aligned}
    \min_{\mb{s}} & \quad\max_{\mathcal{A}, \mathcal{T}_u} \quad\mb{1}^T \mb{B} \mb{A} \mb{s} + p\times \mb{1}^T \mb{B} (\mb{I} - \mb{A}) \mb{\Delta}_{W} \mb{B} \mb{A} \mb{s}\\
    s.t.          & \quad 
    \left\{\begin{array}{lc}
        \mb{1}^T (\mb{s}_0 - \mb{s}) \le \mu,\\
        \mb{0} \le \mb{s} \le \mb{s}_0,\\
        \vert \mathcal{A} \vert \le m,\\
        \vert \mathcal{T}_u \vert \le k,~ \forall u\in \mathcal{A}.
    \end{array}\right.
    \end{aligned}
  \label{eqn:rocpi-relax}
\end{equation}
We next show the convexity of the relaxed network defense problem in \refeq{eqn:rocpi-relax}.
\begin{theorem}
  \label{thm:convex-defense}
  The problem in \refeq{eqn:rocpi-relax} is convex.
\end{theorem}
\begin{proof}
  Let $\mathfrak{A}$ be the set of all possible attacker sets. That is, for any $ |\mathcal{A}| \le m, ~\forall\mathcal{A} \in \mathfrak{A}$. Similarly, let $\mathfrak{T}_u$ be the set of all possible target user set of attacker $u$. That is, for any $|\mathcal{T}_u| \le k, ~\forall \mathcal{T}_u \in \mathfrak{T}_u$. Note that the variable of the inner maximization, i.e., $\mathcal{A}$ and $\mathcal{T}_u$ are only related to $\mb{\Delta}_{W}$. Thus, we also denote $\mb{\Delta}_{W}$ by $\mb{\Delta}_{W}(\mathcal{A}, \mathcal{T}_u)$ for clarity here.

  Note that when $\mathcal{A}$ and $\mathcal{T}_u$ are given, the function
  $$
  \quad\mb{1}^T \mb{B} \mb{A} \mb{s} + p\times \mb{1}^T \mb{B} (\mb{I} - \mb{A}) \mb{\Delta}_{W}(\mathcal{A}, \mathcal{T}_u) \mb{B} \mb{A} \mb{s}
  $$
  is a linear function with respect to the decision variable $\mb{s}$.
  Consequently, the inner maximization of \refeq{eqn:rocpi-relax} is the supremum of all possible linear function of $\mb{s}$, when $\mathcal{A}$ is taken from $\mathfrak{A}$ and $\mathcal{T}_u$ is taken from $\mathfrak{T}_u$. Thus, the objective of the minimization is convex with respect to $\mb{s}$.
  In terms of the constraints, both $\mb{1}^T (\mb{s}_0 - \mb{s}) \le \mu$ and $\mb{0} \le \mb{s} \le \mb{s}_0$ are convex set with respect to $\mb{s}$. Therefore, \refeq{eqn:rocpi-relax} is a convex problem.
\end{proof}

From \refthm{thm:convex-defense}, a direct method to solve \refeq{eqn:rocpi-relax} is to introduce the upper bound $u_b$, and transform the problem in \refeq{eqn:rocpi-relax} as
\begin{equation}
  \begin{aligned}
    \min_{\mb{s}, u_b} & \hspace{20pt}  u_b\\
    s.t.          & \hspace{5pt} 
    \left\{\begin{array}{lc}
        \mb{1}^T (\mb{s}_0 - \mb{s}) \le \mu,\\
        \mb{0} \le \mb{s} \le \mb{s}_0,\\
        \mb{1}^T \mb{B} \mb{A} \mb{s} + p\mb{1}^T \mb{B} (\mb{I} - \mb{A}) \mb{\Delta}_{W}(\mathcal{A}, \mathcal{T}_u) \mb{B} \mb{A} \mb{s} \le u_b,\\
        \forall \mathcal{A} \in \mathfrak{A},\text{ and } \forall \mathcal{T}_u\in \mathfrak{T}_u\\
    \end{array}\right.
    \end{aligned}
  \label{eqn:rocpi-trans1}
\end{equation}
The above transformed problem is a linear programming and can be solved with interior-point method \cite{boyd2004convex}. However, the challenge is that there are too many constraints, because we need to check every possible attacker set $\mathcal{A}$ and $\mathcal{T}_u$.
Therefore, we directly solve \refeq{eqn:rocpi-relax}.

\subsection{Project Subgradient Algorithm for Network Defense}
\begin{algorithm}[t]
  \SetAlgoLined
  \KwIn{$\mb{B}, \mb{W}, \mb{A}, \mb{s}_0$ and constraints of $m, k, \mu$}
  \KwParams{Initial step size $\eta_0$ and number of interation $T_{max}$}
  \KwOut{Optimal controlled innate opinions $\mb{s}^*$}
    $f^* \leftarrow +\infty$, $\mb{s}^* \leftarrow \mb{s}_0$, $\mb{c}_1 = (\mb{1}^T \mb{B} (\mb{I} - \mb{A}))^T$, $T \leftarrow 0$\;
    \While{$T < T_{max}$}{

      {\nonl\texttt{// Attacking step: choose $\mathcal{A}^*$ and $\mathcal{T}_u^*$}}

      $\mb{z}^* \leftarrow \mb{B} \mb{A} \mb{s}$, $\mb{c}_2 \leftarrow \mb{W} \mb{z}^*$\;

      $\mathcal{A}^*, \mathcal{T}_u^* \leftarrow$ \refalg{alg:linear-search}$(\mb{c}_1, \mb{c_2}, z^*, m, k)$\;

      Calculate $\mb{\Delta}_W$ with \refeq{eqn:delta-w} using $\mathcal{A}^*$ and $\mathcal{T}_u^*$\;
     
      $f \leftarrow \mb{1}^T \mb{B} \mb{A} \mb{s} + p\times \mb{1}^T \mb{B} (\mb{I} - \mb{A}) \mb{\Delta}_{W} \mb{B} \mb{A} \mb{s}$\;

      \If{$f^* < f$}{
        $f^* \leftarrow f$, $\mb{s}^* \leftarrow \mb{s}$\;
      }

      {\nonl\texttt{// Defense adjustment: update $\mb{s}$}}

      Calculate subgradient $\mb{g}$ with \refeq{eqn:sub-grad}\;

      $\eta=\eta_0 / \sqrt{k}$\;

      Update $\mb{s} \leftarrow \text{Proj}(\mb{s} - \eta \cdot \mb{g})$\;

      $T \leftarrow T + 1$\;
    }
    \caption{Projected subgradient method for \refeq{eqn:rocpi-relax}}
  \label{alg:subgrad-defense}
\end{algorithm}
The challenge of directly solving \refeq{eqn:rocpi-relax} is that the inner maximization function introduces non-differentiability. Note that there are also constraints on the variable $\mb{s}$. To address these challenges, we devise a projected subgradient algorithm \cite{boyd2003subgradient} to solve \refeq{eqn:rocpi-relax}. The algorithm is summarized in \refalg{alg:subgrad-defense}. In the next, we elaborate the details of the designed algorithm.

The algorithm starts from the uncontrolled innate opinions $\mb{s}_0$. The core idea is to update the controlled innate opinions $\mb{s}$ iteratively with the guidance of subgradient.
In each iteration $T$, We first calculate the subgradient of the objective function in \refeq{eqn:rocpi-relax}. Thanks to our analysis in \refsec{sec:attack}, we can first efficiently solve the inner maximization with \refalg{alg:linear-search}, when $\mb{s}$ is given. This corresponds to line 3 and line 4 in \refalg{alg:subgrad-defense}.

Then, according to \cite{boyd2003subgradient}, we calculate the subgradient (line 10 in \refalg{alg:subgrad-defense}) as
\begin{equation}
  \mb{g} =  (\mb{B} \mb{A})^T \mb{1} + p \times \big( \mb{B} (\mb{I} - \mb{A}) \mb{\Delta}_{W}\mb{B} \mb{A} \big)^T \mb{1}.
  \label{eqn:sub-grad}
\end{equation}
where $\mb{\Delta}_{W}$ is calculated with the optimal attacker set $\mathcal{A}^*$ and optimal target user sets $\mathcal{T}_u^*$ (line 5 in \refalg{alg:subgrad-defense}).
Here, we also record the optimal variable $\mb{s}^*$ and optimal objective $f^*$ during the iteration (line 6-9 in \refalg{alg:subgrad-defense}), since subgradient algorithm is not a descent method \cite{boyd2003subgradient}.

Finally, with the calculated subgradient $\mb{g}$ in \refeq{eqn:sub-grad}, we update the controlled innate opinions $\mb{s}$ as in the line 12 of \refalg{alg:subgrad-defense}.
We adopt the \textit{nonsummable diminishing rule} \cite{boyd2003subgradient} for the choice of step size $\eta$, which guarantees the convergence and the optimality of the subgradient method.
The discussion of convergence and optimality will be shown \refsec{sec:converge}.
In addition, since we have two constraints on $\mb{s}$ in problem \refeq{eqn:rocpi-relax}, we need to project $\mb{s} - \eta \cdot \mb{g}$ on the feasible set when updating. Here, we denote the projection function (algorithm) as $\text{Proj}(\cdot)$.
Note that, in \refeq{eqn:rocpi-relax}, the first constratint is one sinlge equality constraint on the total control budget, while the second one is a box constraint on the range of controlled innate opinion $\mb{s}$.
Thus, the feasible set of \refeq{eqn:rocpi-relax} is the intersection of a hyperplane and a box area. In this work, we implement the projection algorithm in \cite{maculan2003n} whose time complexity is $\mathcal{O}(n)$. To this end, the time complexity of one single iteration of \refalg{alg:subgrad-defense} is $\mathcal{O}(N \times(\log m + m \log k  + 1))$.

\textbf{Game-Theoretical Interpretation.} The above process iterates until the maximal interation $T_{max}$ is reached. We can also divide one single iteration in \refalg{alg:subgrad-defense} as two part. The first part includes line 3-5, which solves the optimal attacking strategy for the adversarial network perturbation, given the current controlled innate opinion $\mb{s}$. We define this part as the \textit{attacking step}. The second part includes line 10-12, where the defender adjusts the variable $\mb{s}$ according to current optimal attacking strategy $\mathcal{A}^*$ and $\mathcal{T}_u^*$. We define this part as the \textit{defense adjustment}. Consequently, the iterations in \refalg{alg:subgrad-defense} can be regarded as a game played between the adversary and the defender. The adversary and the defender update their actions in turn. The updating criterion for the adversary is \refalg{alg:linear-search}, while the updating policy for the defender is based on the subgradient in \refeq{eqn:sub-grad}. Since \refalg{alg:subgrad-defense} is guaranteed to converge to the optimal solution, the converged results can be regarded as the Nash equilibrium in the above game.

\subsection{Convergence of the Proposed Algorithm}
\label{sec:converge}
Given the network defense algorithm in \refalg{alg:subgrad-defense}, we next analyze the convergence of it, and have the following theorem
\begin{theorem}
  \label{thm:converge}
  Let $f^*$ be the optimal value of problem \refeq{eqn:rocpi-relax}, and $f^*_{alg}$ be the optimal value given by \refalg{alg:subgrad-defense} with $T_{max}$ iterations in total. We have
  \begin{equation}
    \label{eqn:k-max-error}
    f^*_{alg} - f^* \le \frac{\mu^2 + \xi_2 N\cdot \left(1 + 2mkp^2\cdot \left(\frac{1}{\alpha_{min}} - 1\right)^2\right)}{\xi_1}
  \end{equation}
  where $\xi_1 = \sum_{i=1}^{T_{max}} 1 / \sqrt{i}$, $\xi_2 = \sum_{i=1}^{T_{max}} 1/ i$, and $\alpha_{min} = \min_i {\alpha_i}$ is the smallest stubbornness value of all users.
\end{theorem}
\begin{proof}
  see \refapdx{apdx:thm-converge}
\end{proof}

From \refthm{thm:converge}, as $T_{max} \rightarrow \infty$, the right hand side of \refeq{eqn:k-max-error} $\rightarrow 0$, and thus, $f^*_{alg} \rightarrow f^*$. This indicates that when the number of iteration $T_{max}$ is sufficiently large, \refalg{alg:subgrad-defense} converge the optimal solution of problem \refeq{eqn:rocpi-relax}.

\refthm{thm:converge} shows the convergence speed of \refalg{alg:subgrad-defense}. That is, in practice, we can use \refthm{thm:converge} to decide the number of iterations $T_{max}$ for a desired optimality $f^*_{alg} - f^*$.
From \refeq{eqn:k-max-error}, we can see that when the network size is large (i.e., a larger $N$), \refalg{alg:subgrad-defense} needs more iterations to converge to a desired optimality. We can also see that when there are more attackers (i.e., a larger $m$), and when the attackers can influence more target users (i.e., a larger $k$), the convergence of \refalg{alg:subgrad-defense} is also slower. In addition, the convergence of \refeq{alg:subgrad-defense} is more sensitive to the strength of perturbation coefficient $p$ due to the square term $p^2$ in the right hand side of \refeq{eqn:k-max-error}. When the perturbation strength is small (i.e., a smaller $p$), the influence of the adversarial network perturbation is limited, and thus, the convergence speed of \refalg{alg:subgrad-defense} is faster. We can also see that user's stubbonrness value plays a critical role in the convergence speed of \refalg{alg:subgrad-defense}. When users have a larger stubbornness value (i.e., a larger $\alpha_{min}$), they are less influenced by others, and the influence of network perturbation is also smaller. Therefore, \refalg{alg:subgrad-defense} can have a faster convergence speed.

\section{Experiments}
\label{sec:exp}
\begin{table}[!t]
  \centering
  \caption{Network Statistics}
  \begin{tabular}{cccc}
    \toprule
    Network & $Nodes: |\mathcal{V}|$ & $Edges |\mathcal{E}|$ & Avg. Degree \\
    \midrule
    Reddit \cite{de2014learning}          & 553   & 8969  & 32.4\\
    Collaboration \cite{tang2009social}   & 679   & 1687  & 4.97\\
    Facebook \cite{leskovec2012learning}  & 4039  & 88234 & 43.6\\
    \bottomrule
  \end{tabular}
  \label{tab:net-stat}
\end{table}
In this section, we run simulation on real social networks to validate our analysis and the proposed control algorithm. We use the network of online discussion forum Reddit (Reddit network for short) in \cite{de2014learning}, the collaboration network of ``Data Mining'' (Collaboration network for short) in \cite{tang2009social}, and the online social network of Facebook (Facebook network for short) in \cite{leskovec2012learning}. The basic information of the above networks are listed in \reftab{tab:net-stat}.
For the Collaboration network and the Facebook network, we randomly generated the users' initial internal opinions $\mb{s}_0$ in range $(0.6, 1)$. For the Reddit network, we use the initial innate opinions that are provided in the original dataset \cite{de2014learning}. For all networks, we randomly generate users' stubbonrness value $\alpha_i$ in range $(0, 1)$. We also randomly generate the influence weights of the adjacency matrix $\mb{W}$, and unify each row of $\mb{W}$ so that it is a row stochastic matrix.
To avoid influence of randomness, we repeat each experiment for 10 times, and report the mean results in the following.

\begin{figure}[t]
  \centering
  \hspace{-8pt}
  \subfigure[Reddit, $k=100$]{\includegraphics[width=0.24\textwidth]{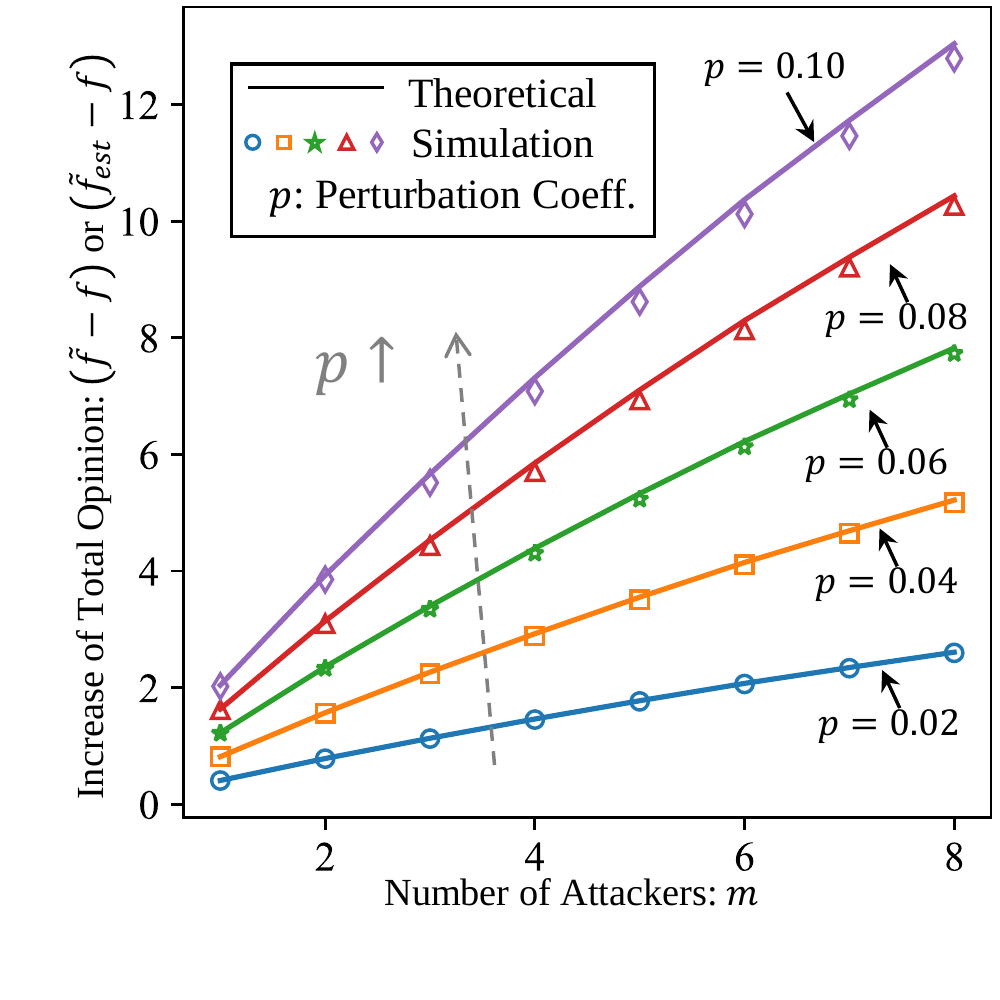}}\hspace{-3pt}
  \subfigure[Reddit, $m = 8$]{\includegraphics[width=0.24\textwidth]{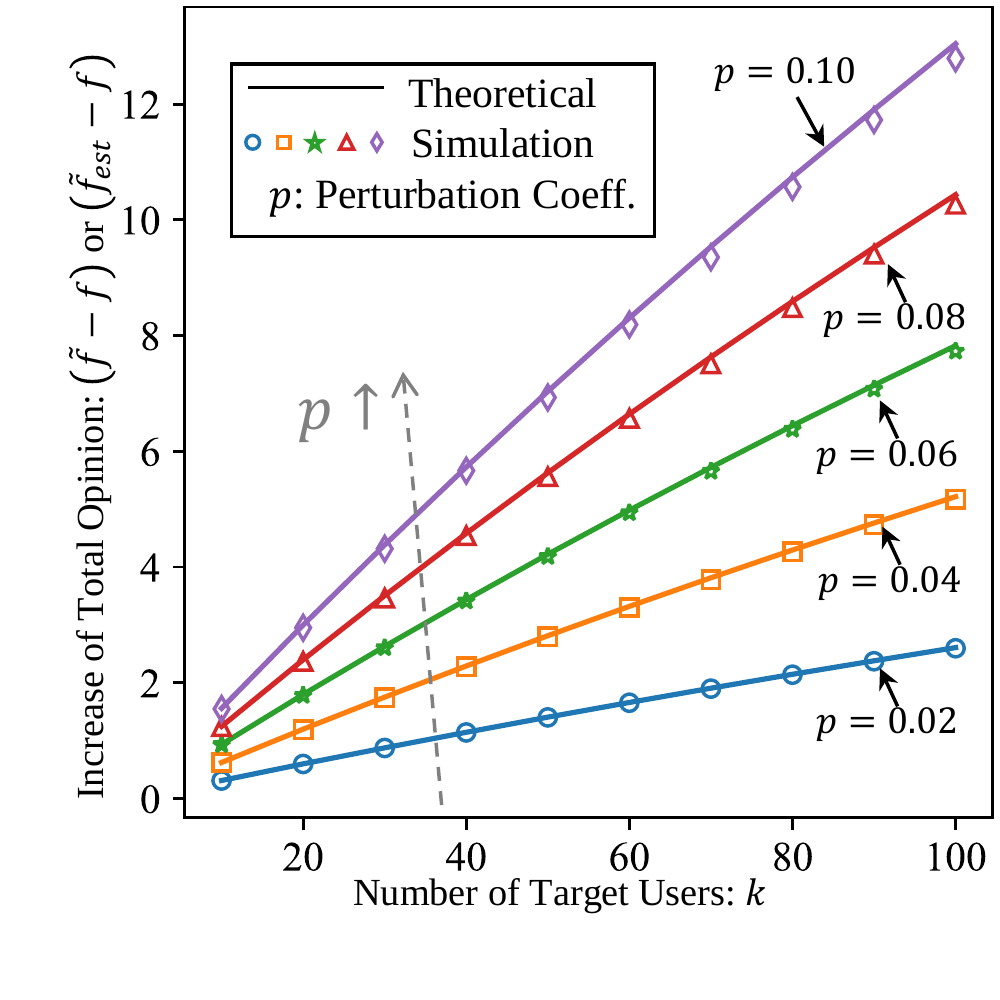}}\\
  \hspace{-8pt}
  \subfigure[Collaboration, $k=100$]{\includegraphics[width=0.24\textwidth]{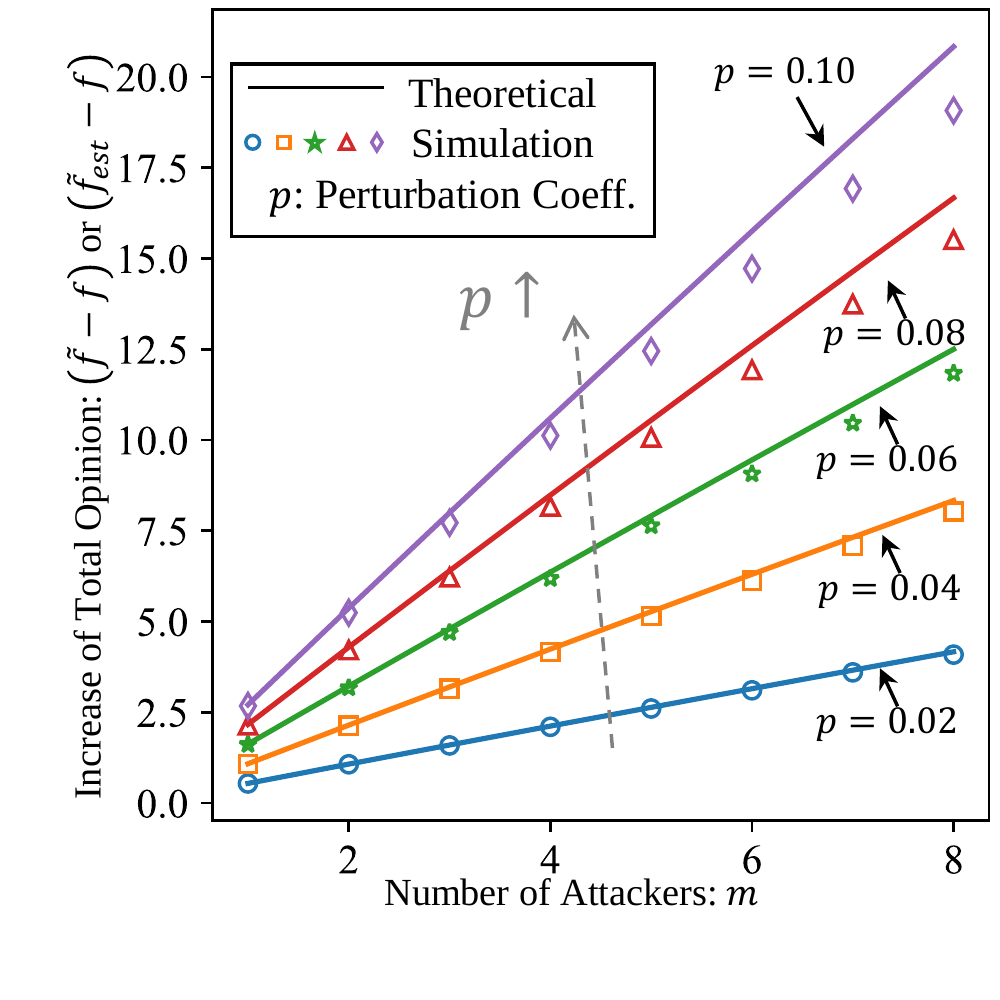}}\hspace{-3pt}
  \subfigure[Collaboration, $m = 8$]{\includegraphics[width=0.24\textwidth]{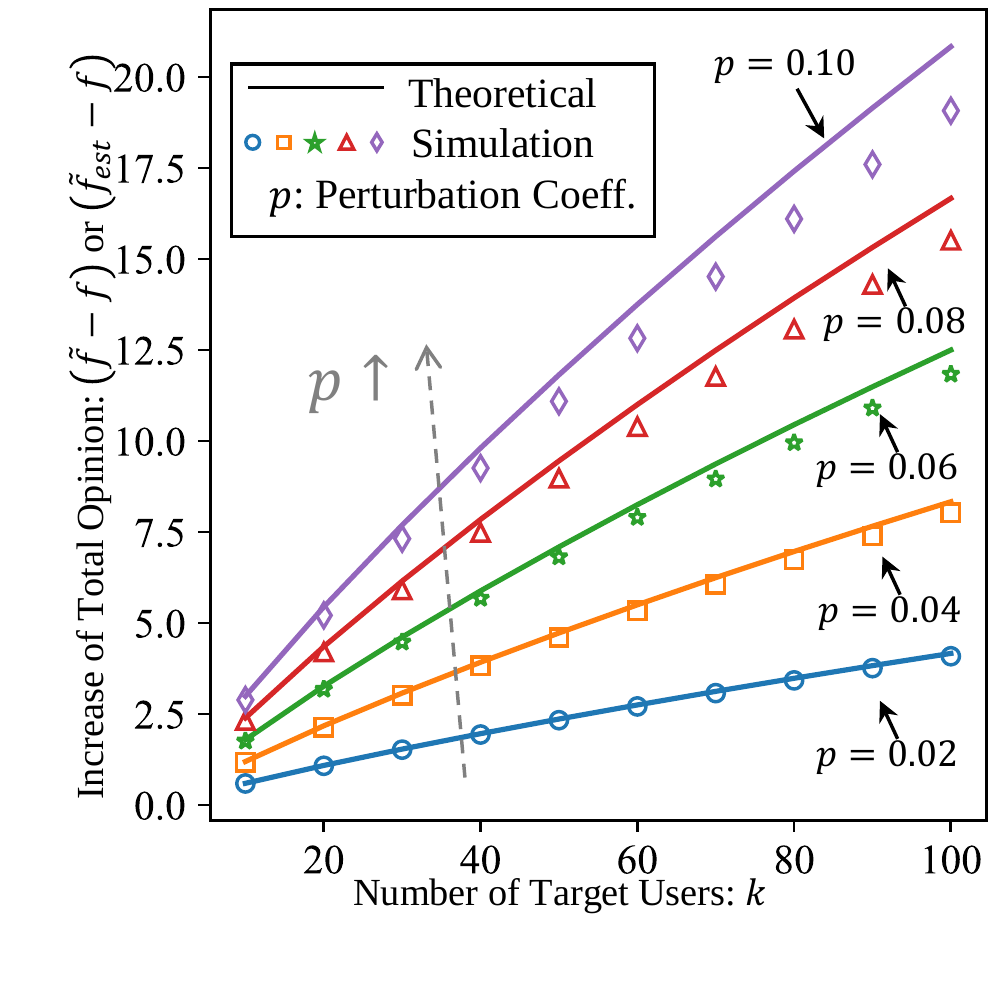}}\\
  \hspace{-8pt}
  \subfigure[Facebook, $k=100$]{\includegraphics[width=0.24\textwidth]{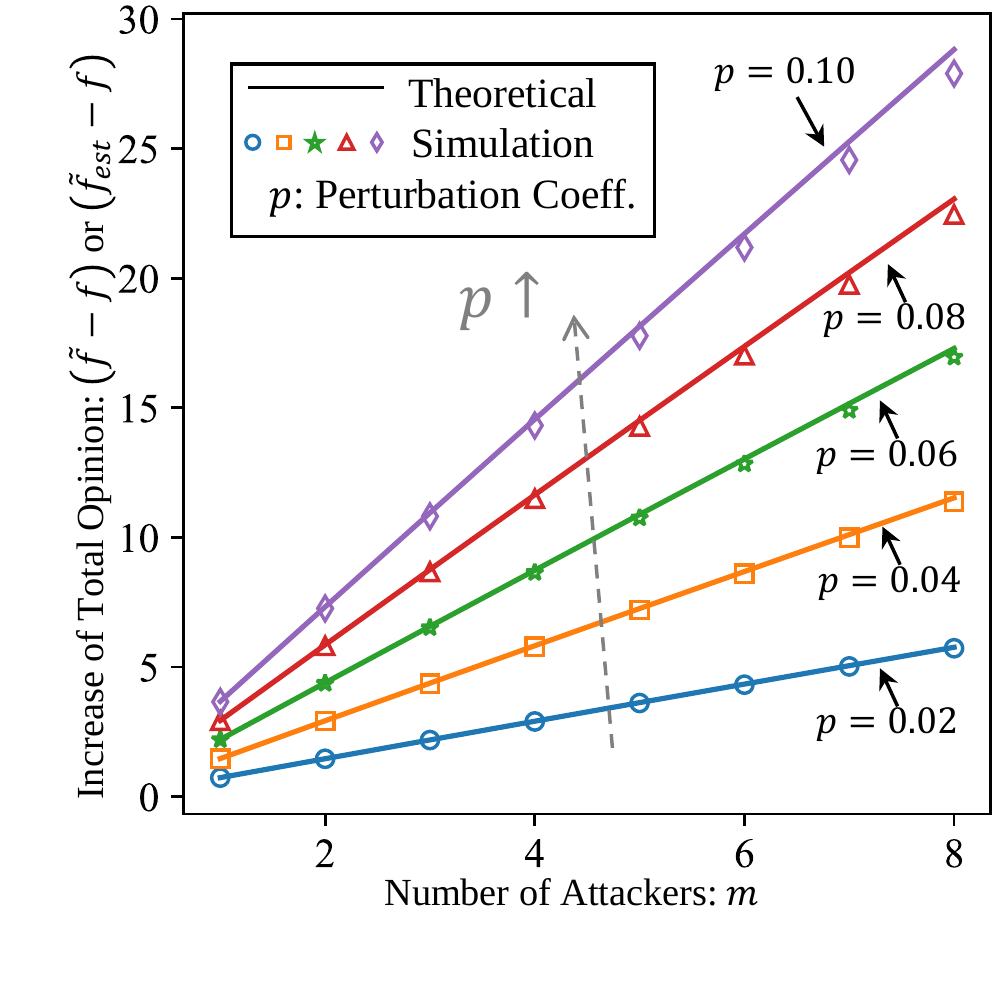}}\hspace{-3pt}
  \subfigure[Facebook, $m = 8$]{\includegraphics[width=0.24\textwidth]{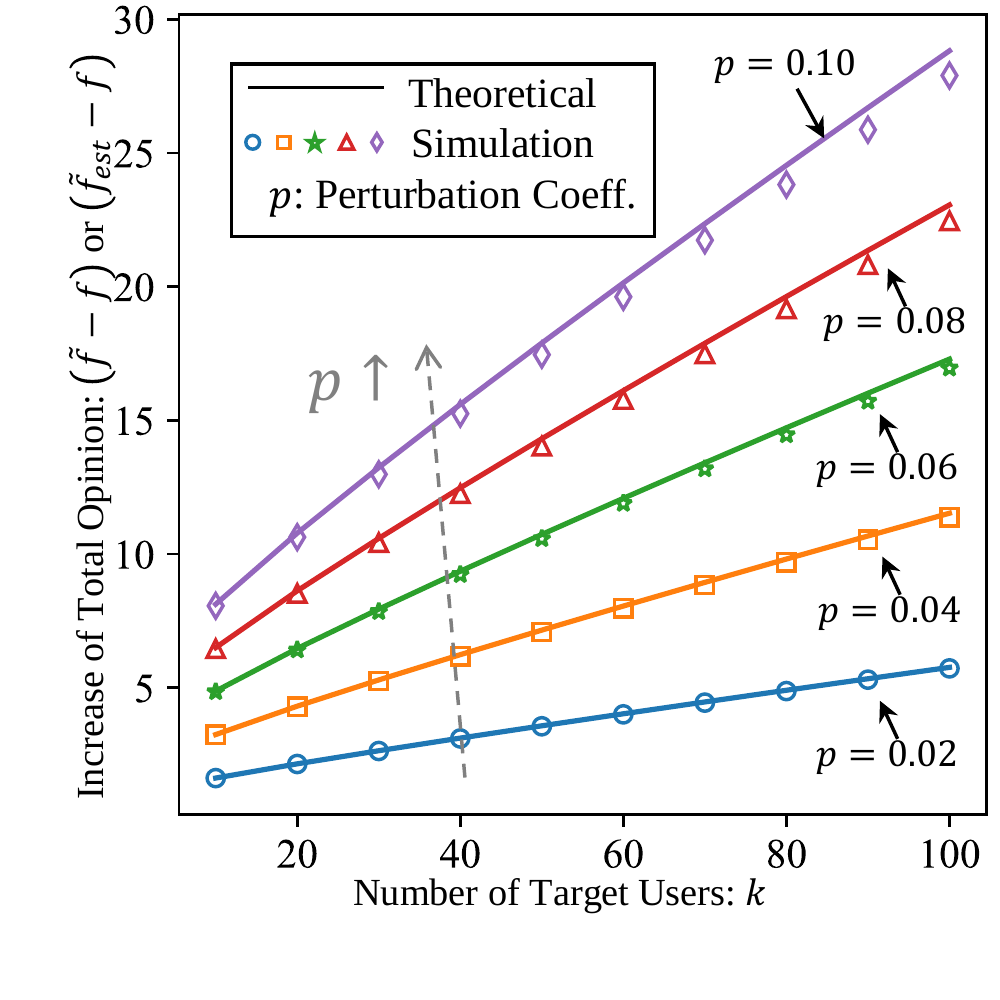}}
  \caption{Increase of total opinion on different social networks. (Left) Number of target users is $k=100$. (Right) Number of attackers is $m = 8$.}
  \label{fig:increase}
\end{figure}
\subsection{Validation of the Approximation in \refprop{prop:approx-new-total}}
As the main results of this paper are based on the approximation in \refprop{prop:approx-new-total}, we first verify the correctness of it.
We first calculate the total opinion $f$ with the initial internal opinions $\mb{s}_0$, users' stubbornness $\alpha_i$, and the adjacency matrix $\mb{W}$ according to \refeq{eqn:fj-equilibrium}. Then, we generate the optimal attackers $\mathcal{A}^*$ and the optimal target users $\mathcal{T}^*_u$ as in \refsec{sec:opt-attack}, and use them to perturb the networks as described in \refsec{sec:net-perturb} to obtain the perturbed network strucutre $\tilde{\mb{W}}$.
Next, with users' initial internal opinions $\mb{s}_0$, stubbornness $\alpha_i$, and the network structure $\tilde{\mb{W}}$, we simulate the FJ opinion dynamics model, and obtain the total opinion with adversarial network perturbtion $\tilde{f}$. We also calculate the approximate total opinion $\tilde{f}_{est}$ as in \refprop{prop:approx-new-total}. We show the increase of the total opinion $(\tilde{f} - f)$ and the estimated one, that is, $(\tilde{f}_{est} - f)$ on three networks in \reffig{fig:increase}.

From \reffig{fig:increase}, we can see that our theoretical approximated results match well with the simulation results on four networks with different perturbation coefficient $p$, the number of attackers $m$, and the number target users $k$. We can see that when there are more attackers and when the attackers can influence more target users, the increase of total opinion is larger. In addition, we can also see that when the perturbation coefficient $p$ is larger, the increase of total opinion is also larger. This observation validates the influence of the perturbation coefficient $p$ on the change of total opinion.

\begin{figure*}[t]
  \centering
  \hspace{-10pt}
  \subfigure[Reddit]{\includegraphics[width=0.335\textwidth]{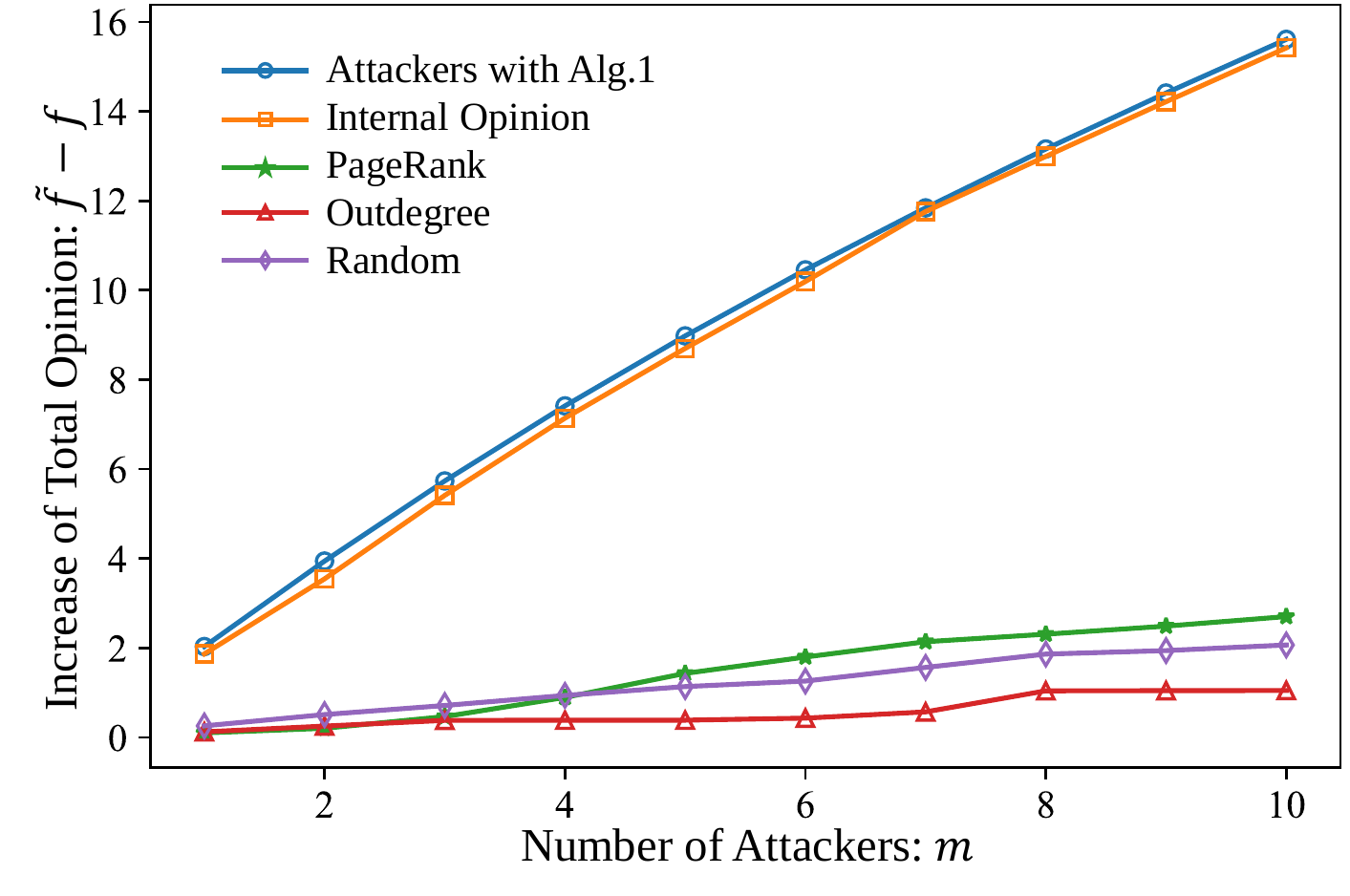}}\hspace{-3pt}
  \subfigure[Collaboration]{\includegraphics[width=0.335\textwidth]{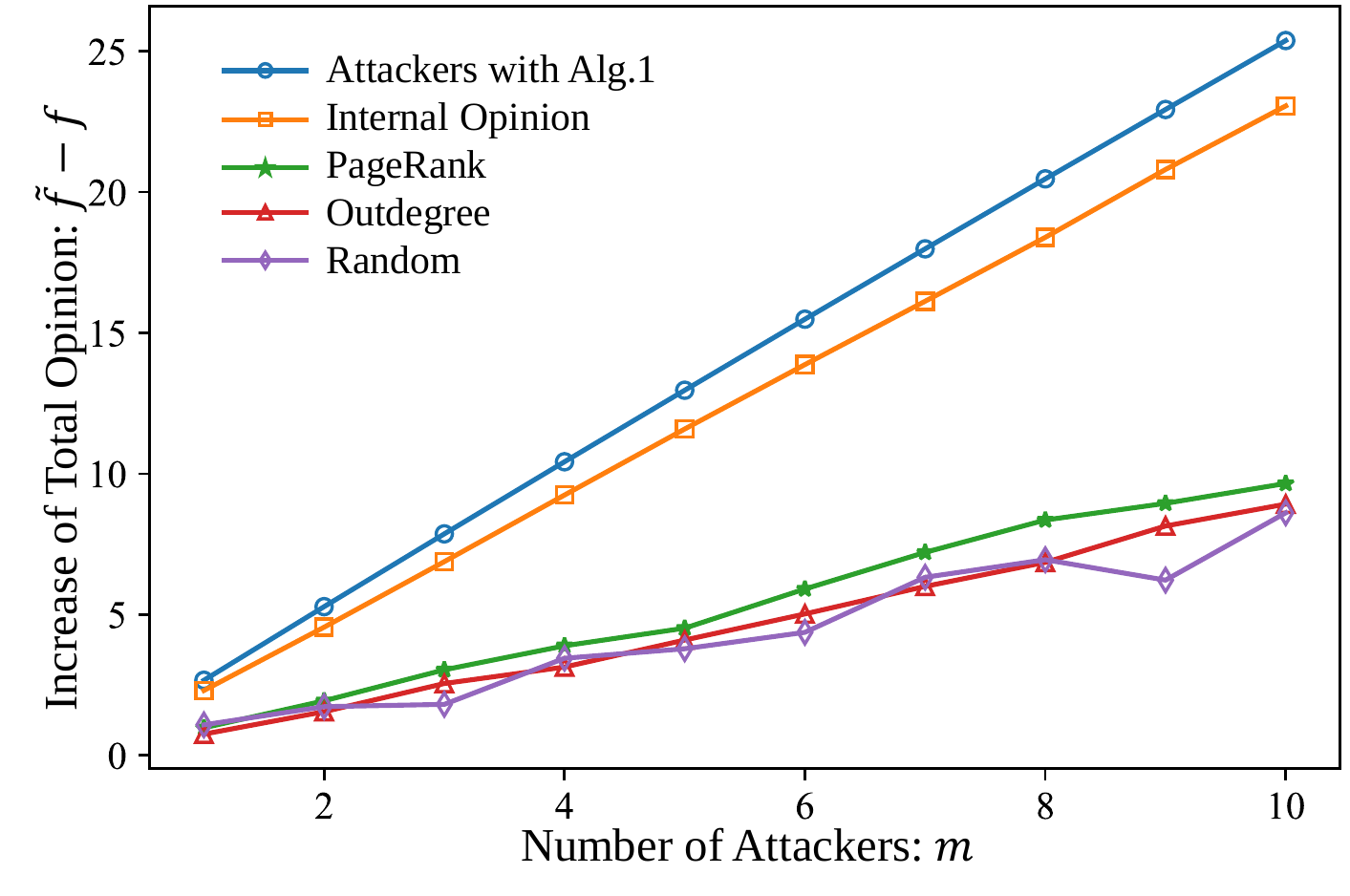}}\hspace{-3pt}
  \subfigure[Facebook]{\includegraphics[width=0.335\textwidth]{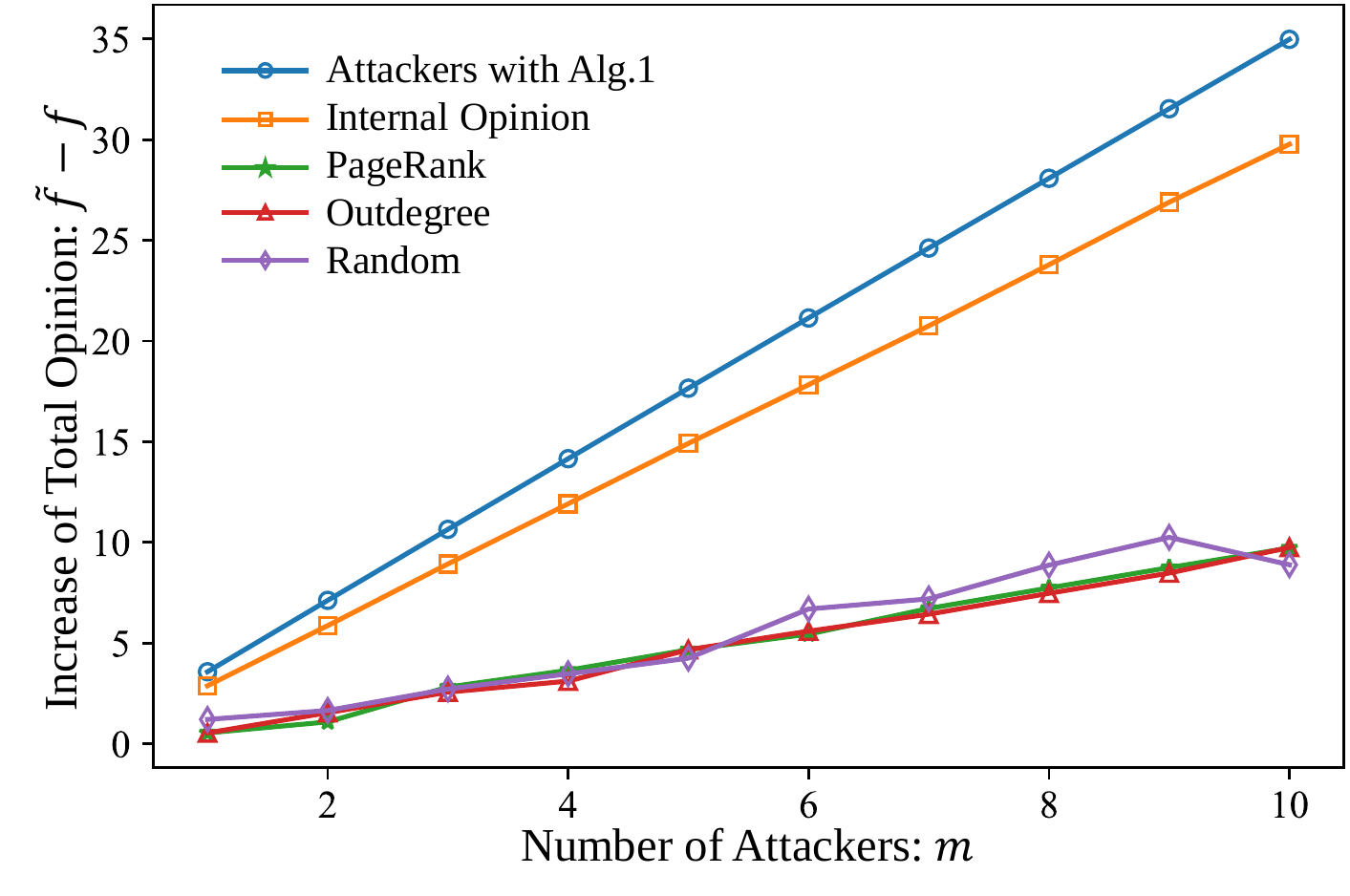}}
  \caption{Increase of total opinion $(\tilde{f} - f)$ when using different criterion to choose attackers.}
  \label{fig:attacker}
\end{figure*}
\begin{figure*}[t]
  \centering
  \hspace{-10pt}
  \subfigure[Reddit]{\includegraphics[width=0.335\textwidth]{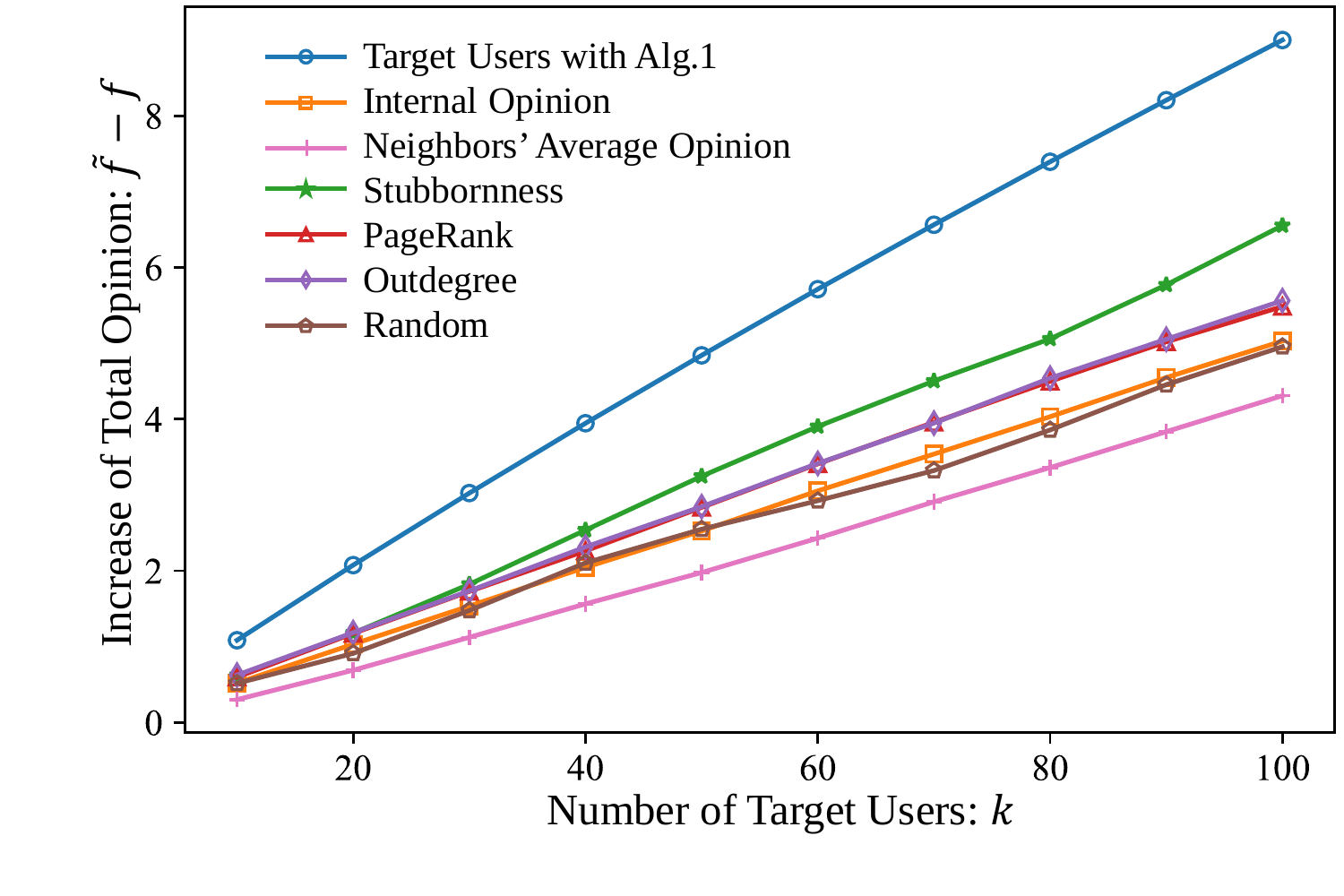}}\hspace{-3pt}
  \subfigure[Collaboration]{\includegraphics[width=0.335\textwidth]{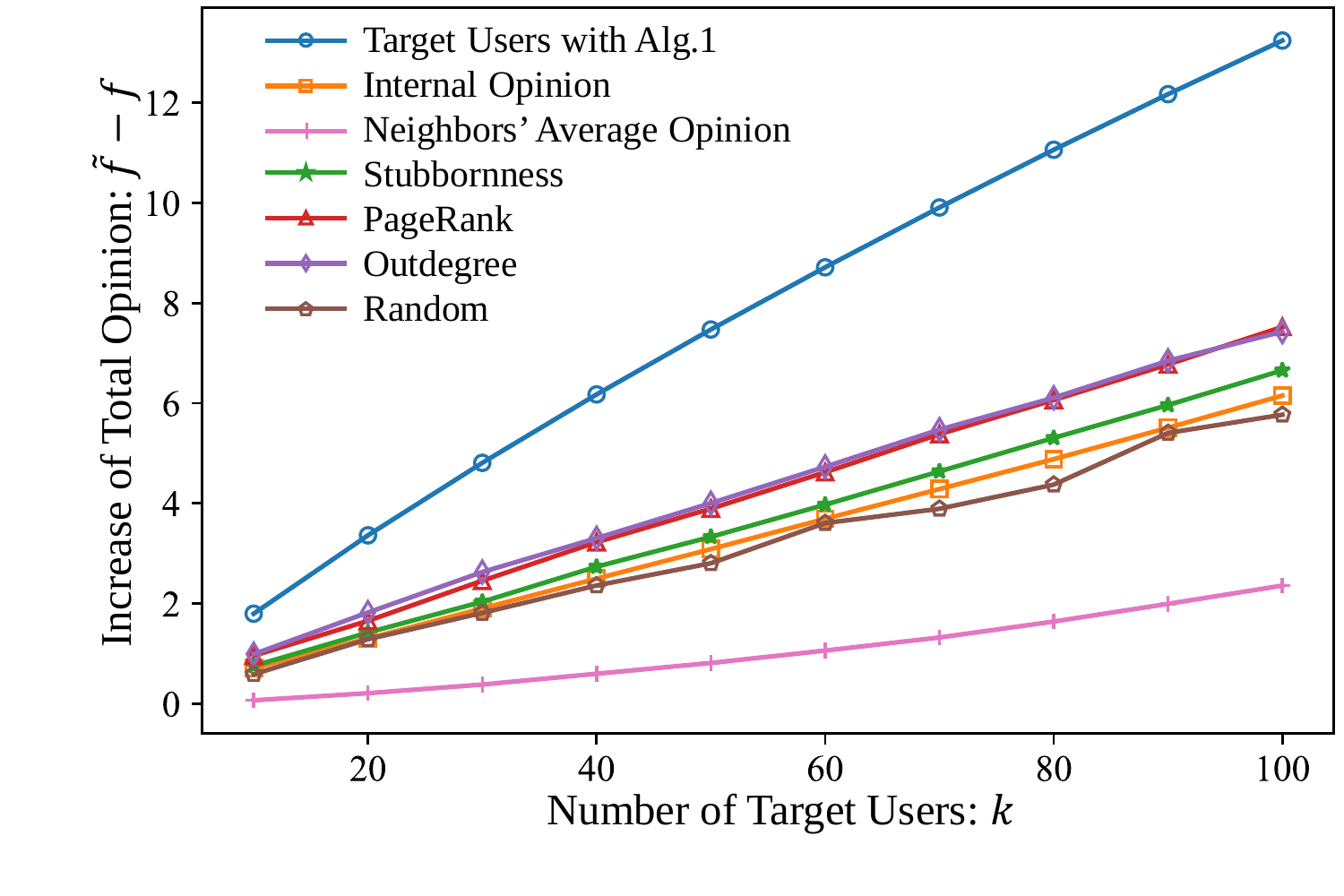}}\hspace{-3pt}
  \subfigure[Facebook]{\includegraphics[width=0.335\textwidth]{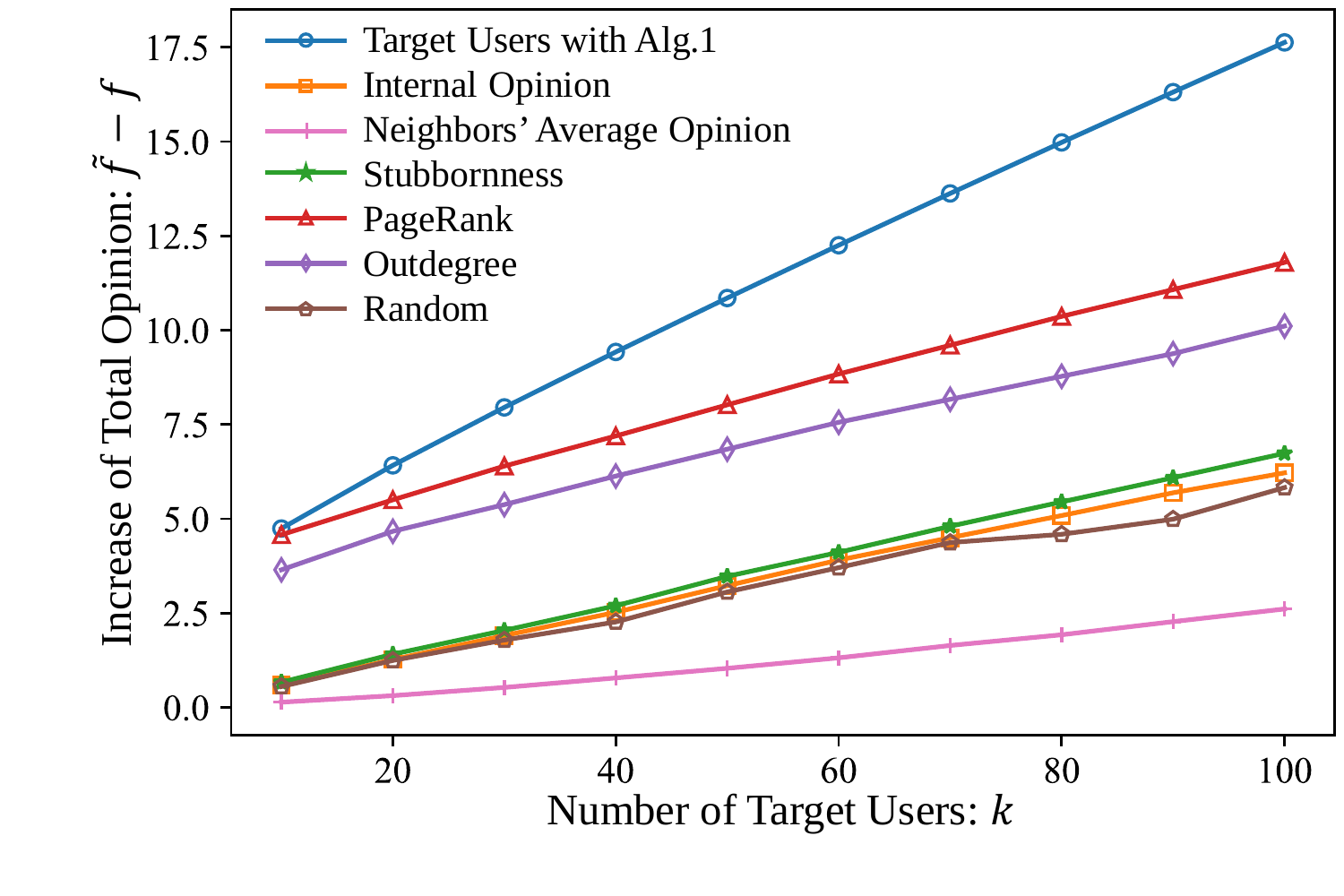}}
  \caption{Increase of total opinion $(\tilde{f} - f)$ when using different criterion to choose target users.}
  \label{fig:target}
\end{figure*}
\subsection{Validation of Optimal Choice of Attackers and Target Users}
In this section, we validate the performance of the selection criterion of attackers and target users in \refsec{sec:attack}.

\subsubsection{Choose of Attackers} We first validate the selection of optimal attackers in \refalg{alg:linear-search}. We choose the following heuristic criterion as baseline.
\begin{itemize}
  \item \textbf{Internal Opinion} which selects users with the largest internal opinion $s_i$ as the attacker;
  \item \textbf{PageRank} which selects users with the largest PageRank value as the attacker;
  \item \textbf{Outdegree} which selects users with the largest outdegree and can influence more users as the attacker; and
  \item \textbf{Random} which randomly selects users as the attackers.
\end{itemize}
For each selected attacker $u$, we use the proposed selection rule of target user in \refsec{sec:opt-attack} to decide his/her optimal target user set $\mathcal{T}^*_u$.
The perturbation coefficient is set to $p = 0.1$, and the number of target user for each attacker is set to $k=100$. We change the number of attacker $m$, and show the increase of total opinion $\tilde{f} - f$ on three networks in \reffig{fig:attacker}. When we select more attackers (i.e., a larger $m$), the increase of total opinion $(\tilde{f} - f)$ becomes larger. This shows that more attackers can cause a large perturbation in terms of the total opinion.

We also can see that the proposed selection criterion of the attackers outperforms other heuristics criterion by a large extent. We empirically find that choosing users with large internal opinion performs better than other two heuristics that only use network structure information and randomly selection. This indicates that users with larger internal opinions can cause a larger perturbation on the total opinion.

\subsubsection{Choose of Target Users} Next, we validate the selection of optimal target users in \refalg{alg:linear-search}. We also adopt the four heuristic criterions in the above as the baseline methods to select the target users for each attacker. In addition, according to our analysis in \refsec{sec:attack}, we add the following two heuristic criterions:
\begin{itemize}
  \item \textbf{Stubbornness} which selects users with the \textit{smallest} stubbornness value as target user;
  \item \textbf{Neighbors' Average Opinion} which selects users with the largest neighbors' average opinion $(\mb{W} \mb{z}^*)_i$ as target user.
\end{itemize}
We select $m=5$ attackers according to the criterion as in \refsec{sec:attack}, and set the strength of perturbation coefficient $p=0.1$. The increase of total opinion $(\tilde{f} - f)$ with different number of target user $k$ is shown in \reffig{fig:target}

From \reffig{fig:target}, we can see that the selection criterion of target user in \refalg{alg:linear-search} performs better than all heuristic baseline methods. When the attackers can influence more target users (i.e., a larger $k$), the their influence on the increase of total opinion is larger. In addition, we can also see that, the stubbornness criterion performs better than other heuristics methods on Reddit network. However, on Citation network and Facebook network, heuristic criterions related to network structure (PageRank critertion and Outdegree criterion) performs better than other baselines. This indicates that the network structure plays a more critical role on the perturbation of these two networks.


\begin{figure*}[t]
  \centering
  \hspace{-10pt}
  \subfigure[Reddit]{\includegraphics[width=0.335\textwidth]{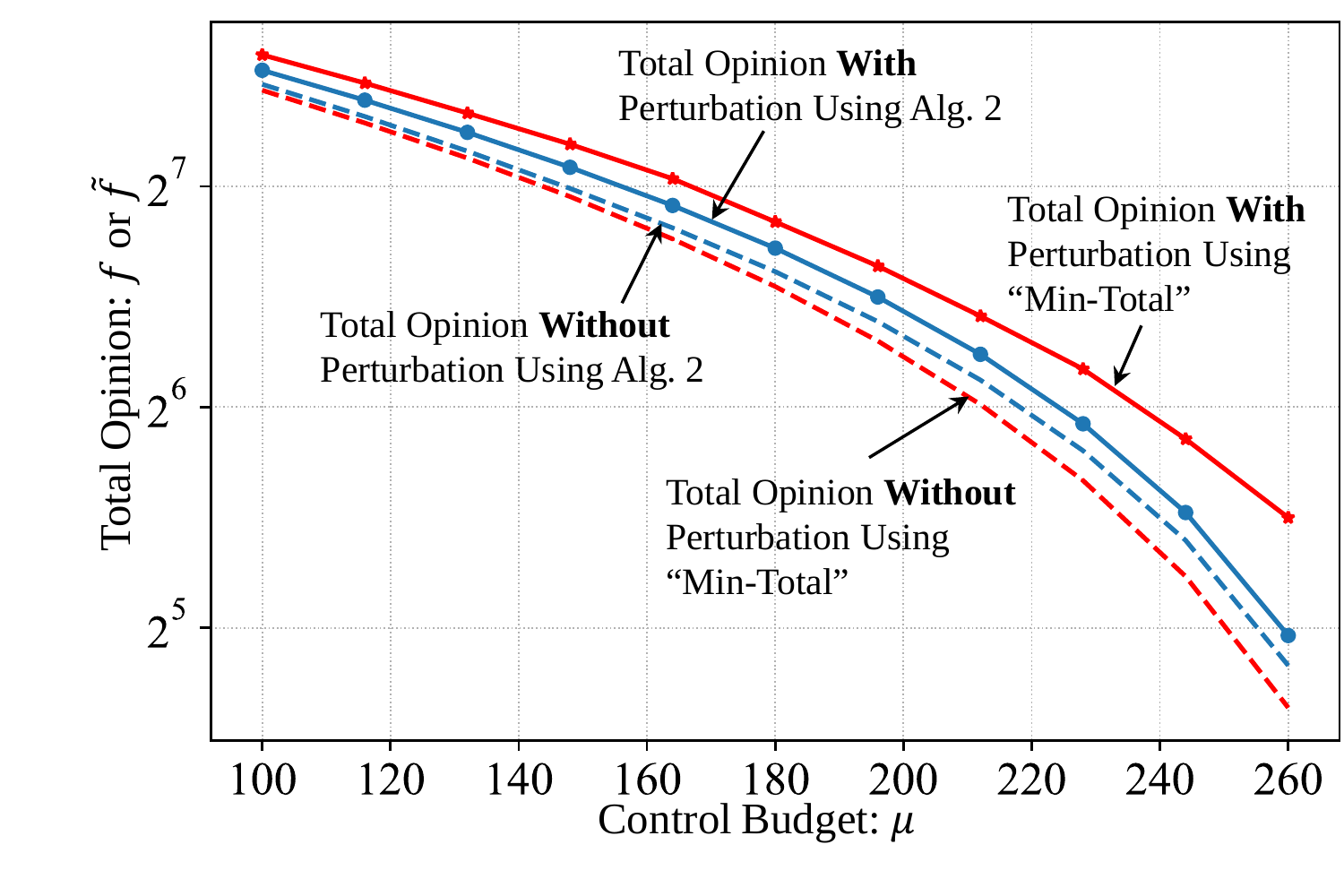}}\hspace{-3pt}
  \subfigure[Collaboration]{\includegraphics[width=0.335\textwidth]{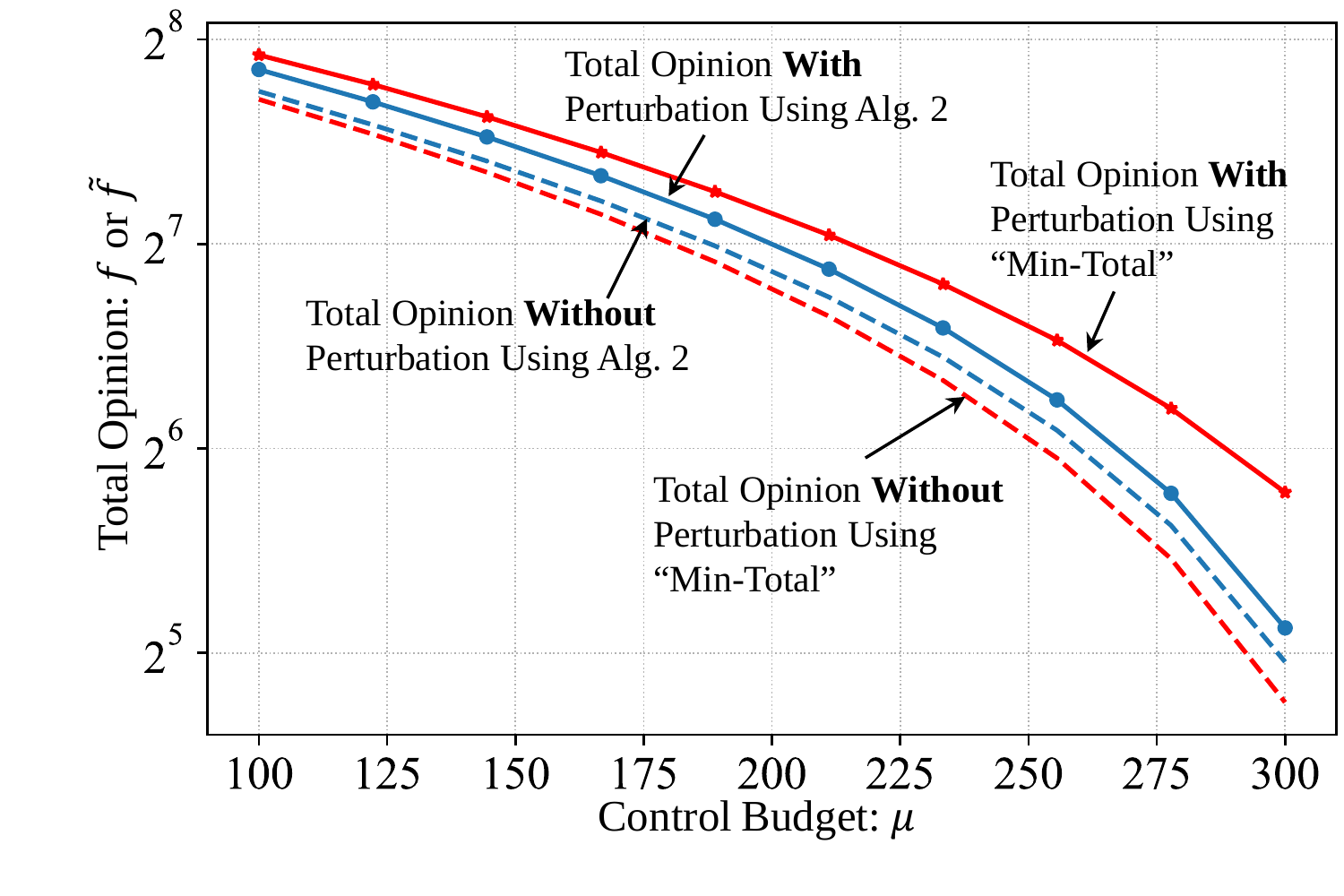}}\hspace{-3pt}
  \subfigure[Facebook]{\includegraphics[width=0.335\textwidth]{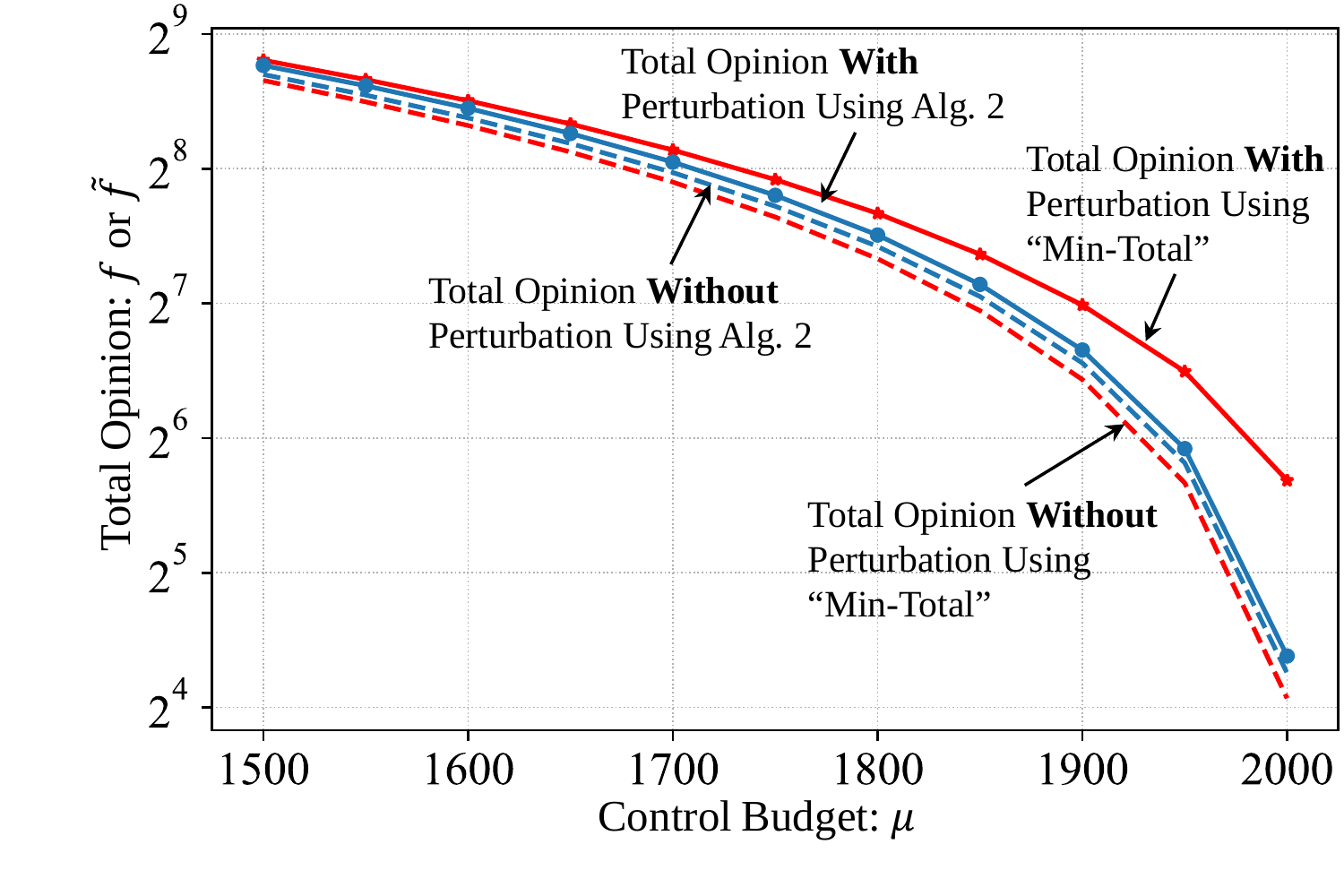}}
  \caption{The total opinion with and without network perturbation using different control strategy.}
  \label{fig:defense}
\end{figure*}
\subsection{Validation of Network Defense}
In this secion, we validate the proposed network defense algorithm in \refsec{sec:defense}. We set the perturbation coefficient $p=0.15$. The number of attackers and the number of target users for each attacker are set to $m=6$ and $k=100$, respectively, for three networks. We can observe the similar results with other parameters settings, and omit them. With the initial internal opinions $\mb{s}_0$, we first use the proposed algorithm in \refsec{sec:defense} to obtain the controlled internal opinion $\mb{s}$. The number of iteration $T_{max}$ is calculated with \refeq{eqn:k-max-error} by setting the error tolerence to $0.01$. Then, we use \refeq{eqn:ori-total-opinion} to obtain the total opinion $f$ when there is no adversarial network perturbation. Next, we perturb the network structure as described in \refsec{sec:attack}, and simulate the FJ opinion dynamics to obtain the total opinion $\tilde{f}$ after the adversarial network perturbation. In \reffig{fig:defense}, We plot the total opinion with and without the adversarial network perturbation in \reffig{fig:defense}, respectively, when the control budget $\mu$ changes. We also compare the proposed defense algorithm with the ``Min-Total'' control strategy as in \cite{xu2020opinion}. That is, We use ``Min-Total'' algorithm to derive the controlled internal opinions, and repeat the above process to obtain the total opinion with and without the adversarial network perturbation. The results are also shown in \reffig{fig:defense}.

From \reffig{fig:defense}, when there is no network perturbation, the total opinion given by ``Min-Total'' is slightly lower than that given by our defense algorithm. This is because that ``Min-Total'' focus on minimizing the total opinion when there is no network perturbation, and can obtain the optimal solution \cite{xu2020opinion}.
However, when the adversarial network perturbation exists, the total opinion given by the proposed network defense algorithm is significantly lower than that given by ``Min-Total'' algorithm.

When using the ``Min-Total'' control algorithm, we can see that the total opinion increases dramatically after the network perturbation. Furthermore, the increase in the total opinion becomes larger when the control budget $\mu$ is larger. For the Facebook network, when the control budget is $\mu=2000$, the total opinion surprisingly increases 200\% compared to the one without adversarial network perturbation, even if there is only six attackers and each attackers can only influence 2.5\% users in the whole network. Nevertheless, when using the proposed network defense algorithm, the increase of total opinion is negaliable. These observations show that the proposed network defense algorithm is more robust to the adversarial network perturbation, and thus, can obtain a lower total opinion with the perturbation.


\section{Conclusion}
\label{sec:conclusion}
In this work, we consider the adversarial network perturbation, where the adversary can let some attackers spread their extreme opinions to target users. We theoreticaly analyze such adversarial network perturbation's influence on the network's total opinion. From the adversary's perspective, we analyze the optimal strategy to choose the attackers and the target users, so that the total opinion is maximized. Then, from the network defender's perspective, we formualte a Stackelberg game to minimize the total opinion under such adversarial network perturbation, and device an projected subgradient algorithm to solve the fromulated game. Simulations on real social networks validate our analysis of the network perturbation and the effectiveness of the proposed opinion contorl algorithm.

\ifCLASSOPTIONcaptionsoff
  \newpage
\fi

\bibliographystyle{IEEEtran}
\bibliography{journal_main}

\appendices
\section{Proof of \refthm{thm:converge}}
\label{apdx:thm-converge}
To prove \refthm{thm:converge}, we first present the following lemma.

\begin{lemma}[Convergence of Project Subgradient Method]
  \label{lma:converge-subgradient}
  Let $\mb{s}^*$ be the optimal solution to the problem in \refeq{eqn:rocpi-relax}, and $\mb{s}_0$ is the initial point of \refalg{alg:subgrad-defense}. If $\Vert \mb{s}_0 - \mb{s}\Vert_2 \le R$ and the subgradient in \refeq{eqn:sub-grad} satisfies that $\Vert \mb{g} \Vert_2 \le G$, then we have
  \begin{equation}
    f^*_{alg} - f^* \le \frac{R^2 + \eta_0^2G^2\xi_2}{2\eta_0\xi_1},
    \label{eqn:basic-converge}
  \end{equation}
  where $f^*_{alg}$ is the optimal value given by \refalg{alg:subgrad-defense} with $T_{max}$ iterations, $\eta_0$ is the initial step size of \refalg{alg:subgrad-defense}, $\xi_1 = \sum_{i=1}^{T_{max}} 1 / \sqrt{i}$, and $\xi_2 = \sum_{i=1}^{T_{max}} 1/ i$.
\end{lemma}
\begin{proof}
  See Section 3.2 and Section 6 in \cite{boyd2003subgradient}.
\end{proof}

Now, with \reflma{lma:converge-subgradient}, to prove \refthm{thm:converge}, we need to find the upper bound $R^2$ and $G^2$. We first calculate $R^2$, which is the upper bound of the Euclidean distance between the initial innate opinions $\mb{s}_0$ and the optimal controlled innate opinion $\mb{s}$. Consequently, we have the following lemma
\begin{lemma}
  \label{lma:up-s}
 \begin{equation}
   \Vert \mb{s}_0 - \mb{s}^* \Vert_2^2 \le \mu^2,
 \end{equation}
 where $\mu$ is the control budget.
\end{lemma}
\begin{proof}
   Let $\mb{d} = \mb{s}_0 - \mb{s}^*$, and $d_i = s_0(i) - s^*_i$. Since $\mb{s}^*$ is the optimal solution to problem \refeq{eqn:rocpi-relax}, it is also a feasible solution. Thus, $\mb{0} \preceq  \mb{s}^* \preceq \mb{s}_0$, and $\mb{0} \preceq \mb{d} \preceq \mb{s}_0$, that is, $0 \le d_i \le s_0(i)$. In addition, with the control budget constraint $\mb{1}^T \mb{s}_0 - \mb{1}^T \mb{s}^* = \mu$, we also have $\mb{1}^T \mb{d} = \mu$, that is, $\sum_i d_i = \mu$. Note that
   $$
    \begin{aligned}
      \mu^2 &= (d_1 + \cdots + d_N)^2 = \sum_i d_i^2 + 2\times \sum_{i\neq j} d_i d_j\\
      &= \Vert \mb{d} \Vert_2^2 + 2\times \sum_{i\neq j} d_i d_j\\
      \Rightarrow & \Vert \mb{d} \Vert_2^2 = \Vert \mb{s}_0 - \mb{s}^* \Vert_2^2 = \mu^2 - 2\times \sum_{i\neq j} d_i d_j.
    \end{aligned}
   $$
   With $d_i \ge 0$ for each $i$, we have $2\times \sum_{i\neq j} d_i d_j \ge 0$. Consequently, we have$\Vert \mb{d} \Vert_2^2 \le \mu^2$. This ends the proof.
\end{proof}
Next, we analyze $G$, which is the upper bound of subgradient $\mb{g}$'s length, and have the following lemma
\begin{lemma}
  \label{lma:up-g}
  \begin{equation}
    \Vert \mb{g} \Vert_2^2 \le N \| \mb{B} \|_2^2 \left(1 + p^2mk\cdot\| \mb{B} \|_2^2 (1 + \| \mb{W} \|_2^2) \right).
  \end{equation}
  where $\mb{W}$ is the adjacency matrix, and $\mb{B} = [\mb{I} - (\mb{I} - \mb{A}) \mb{W}]^{-1}$.
\end{lemma}
\begin{proof}
  With the definition of subgradient in \refeq{eqn:sub-grad}, we first have
  $$
  \begin{aligned}
    \Vert \mb{g} \Vert_2^2 &= \Vert (\mb{B} \mb{A})^T \mb{1} + p \times \big( \mb{B} (\mb{I} - \mb{A}) \mb{\Delta}_{W}\mb{B} \mb{A} \big)^T \mb{1} \Vert_2^2\\
    &\le \Vert (\mb{B} \mb{A})^T \mb{1} \Vert_2^2 + p^2 \times \Vert \big( \mb{B} (\mb{I} - \mb{A}) \mb{\Delta}_{W}\mb{B} \mb{A} \big)^T \mb{1} \Vert_2^2.
  \end{aligned}
  $$
  We first analyze the first term $\Vert (\mb{B} \mb{A})^T \mb{1} \Vert_2^2$, and we can have
  $$
  \begin{aligned}
    \Vert (\mb{B} \mb{A})^T \mb{1} \Vert_2^2 &\le \| (\mb{B} \mb{A})^T \|_2^2 \cdot \| \mb{1} \|_2^2 = N \cdot \| \mb{B} \mb{A} \|_2^2 \\
    &\le N \cdot \| \mb{B} \|_2^2 \| \mb{A} \|_2^2 = \alpha_{\max}^2 N \| \mb{B} \|_2^2 \le N \| \mb{B} \|_2^2,
  \end{aligned}
  $$
  where $\alpha_{\max}$ is the largest stubbornness of the whole population and $\alpha_{\max} \le 1$.  

  Next, we analyze $\Vert \big( \mb{B} (\mb{I} - \mb{A}) \mb{\Delta}_{W}\mb{B} \mb{A} \big)^T \mb{1} \Vert_2^2$. We first rewrite this term as
  \begin{equation}
    \label{eqn:appendix-derive-1}
    \begin{aligned}
      &\Vert \big( \mb{B} (\mb{I} - \mb{A}) \mb{\Delta}_{W}\mb{B} \mb{A} \big)^T \mb{1} \Vert_2^2 = \Vert \mb{A} \mb{B}^T \mb{\Delta}_{W}^T (\mb{I} - \mb{A}) \mb{B}^T \mb{1} \Vert_2^2\\
      &\le \| \mb{A} \|_2^2 \cdot  \| \mb{B}^T \|_2^4 \cdot \| \mb{\Delta}_{W}^T \|_2^2 \cdot \| \mb{I} - \mb{A} \|_2^2 \cdot \| \mb{1} \|_2^2\\
      &= N \alpha_{\max}^2 \cdot (1 - \alpha_{\min})^2 \| \mb{B} \|_2^4 \cdot \| \mb{\Delta}_{W} \|_2^2\\
      & \le N \| \mb{B} \|_2^4 \cdot \| \mb{\Delta}_{W} \|_2^2
    \end{aligned}
  \end{equation}
  where $\alpha_{\min}$ is the smallest stubbornness of the whole population, and $1 - \alpha_{\min} \le 1$.
  
  Next, we anlayze the term $\Vert \mb{\Delta}_{W} \Vert_2^2$. With the definition of $\mb{\Delta}_{W}$ in \refeq{eqn:delta-w}, we have
  $$
  \begin{aligned}
    \Vert \mb{\Delta}_{W} \Vert_2^2 &= \left\Vert \sum_{v \in \mathcal{T}} \mb{e}_v \sum_{u \in \mathcal{A}_v} \mb{e}_u^T - \sum_{v \in \mathcal{T}} \vert \mathcal{A}_v\vert \mb{e}_v \mb{e}_v^T\mb{W} \right\Vert_2^2\\
    &= \left\Vert \sum_{v \in \mathcal{T}}\sum_{u \in \mathcal{A}_v} \mb{e}_v \mb{e}_u^T - \sum_{v \in \mathcal{T}} \sum_{u\in\mathcal{A}_v} \mb{e}_v \mb{e}_v^T\mb{W} \right\Vert_2^2\\
    &= \left\Vert \sum_{v \in \mathcal{T}}\sum_{u \in \mathcal{A}_v} \left(\mb{e}_v \mb{e}_u^T - \mb{e}_v \mb{e}_v^T\mb{W}\right) \right\Vert_2^2\\
    &= \left\Vert \sum_{u \in \mathcal{A}}\sum_{v \in \mathcal{T}_u} \left(\mb{e}_v \mb{e}_u^T - \mb{e}_v \mb{e}_v^T\mb{W}\right) \right\Vert_2^2\\
    &\le \sum_{u \in \mathcal{A}}\sum_{v \in \mathcal{T}_u} \left\Vert \mb{e}_v \mb{e}_u^T - \mb{e}_v \mb{e}_v^T\mb{W} \right\Vert_2^2\\
    &\le \sum_{u \in \mathcal{A}}\sum_{v \in \mathcal{T}_u} \left\Vert \mb{e}_v \mb{e}_u^T \right\Vert_2^2 + \left\Vert \mb{e}_v \mb{e}_v^T \right\Vert_2^2 \cdot \| \mb{W} \|_2^2\\
  \end{aligned}
  $$
  For the term $\mb{e}_v \mb{e}_u^T$, its Frobenius norm satisfies $\Vert \mb{e}_v \mb{e}_u^T \Vert^2_F = tr(\mb{e}_v \mb{e}_u^T \mb{e}_u \mb{e}_v^T) = 1$. According to the relation that $\Vert \mb{C} \Vert_2 \le \Vert \mb{C} \Vert_F$ for any square matrix $\mb{C}$, we can have $\Vert \mb{e}_v \mb{e}_u^T \Vert_2 \le \Vert \mb{e}_v \mb{e}_u^T \Vert_F = 1$. Thus, $\Vert \mb{e}_v \mb{e}_u^T \Vert_2^2 \le 1$. Similarly, we can also have $\Vert \mb{e}_v \mb{e}_v^T \Vert_2^2 \le 1$.
  Consequently, we have 
  \begin{equation}
    \begin{aligned}
      \Vert \mb{\Delta}_{W} \Vert_2^2 &\le \sum_{u \in \mathcal{A}}\sum_{v \in \mathcal{T}_u} \left\Vert \mb{e}_v \mb{e}_u^T \right\Vert_2^2 + \left\Vert \mb{e}_v \mb{e}_v^T \right\Vert_2^2 \cdot \| \mb{W} \|_2^2 \\
    &\le \sum_{u \in \mathcal{A}}\sum_{v \in \mathcal{T}_u} (1 + \| \mb{W} \|_2^2) \le mk\cdot(1 + \| \mb{W} \|_2^2).
    \end{aligned}
    \label{eqn:delta-W-upper}
  \end{equation}
  With $\Vert \mb{\Delta}_{W} \Vert_2^2$ in \refeq{eqn:delta-W-upper} and the derivation in \refeq{eqn:appendix-derive-1}, we have
  $$
  \begin{aligned}
    \Vert \big( \mb{B} (\mb{I} - \mb{A}) \mb{\Delta}_{W}\mb{B} \mb{A} \big)^T \mb{1} \Vert_2^2 &\le N \| \mb{B} \|_2^4 \| \mb{\Delta}_{W} \|_2^2 \\
    &\le Nmk \| \mb{B} \|_2^4 \| (1 + \| \mb{W} \|_2^2),
  \end{aligned}
  $$
  Consequently, we have the upper bound of subgradient $\mb{g}$, that is,
  $$
  \begin{aligned}
    \Vert \mb{g} \Vert_2^2 &\le \Vert (\mb{B} \mb{A})^T \mb{1} \Vert_2^2 + p^2 \times \Vert \big( \mb{B} (\mb{I} - \mb{A}) \mb{\Delta}_{W}\mb{B} \mb{A} \big)^T \mb{1} \Vert_2^2\\
    &\le N \| \mb{B} \|_2^2 \left(1 + p^2mk\cdot\| \mb{B} \|_2^2 (1 + \| \mb{W} \|_2^2) \right)
  \end{aligned}
  $$
  This ends the proof.
\end{proof}
With the above lemmas, we can replace $R^2$ and $G^2$ in \refeq{eqn:basic-converge} with $\mu^2$ and $\left(1 + 2mkp^2\cdot \left(\frac{1}{\alpha_{min}} - 1\right)^2\right)$, respectively. This leads to the convergence in \refeq{eqn:k-max-error}.

\end{document}